\documentclass{article}

\usepackage[utf8]{inputenc} \usepackage[T1]{fontenc}    \usepackage{hyperref}       \usepackage{url}            \usepackage{booktabs}       \usepackage{amsfonts}       \usepackage{nicefrac}       \usepackage{microtype}      \usepackage{xcolor}         \usepackage[a4paper, top=2.5cm, bottom=2.5cm, left=3cm, right=3cm]{geometry}

\usepackage{amsmath,amsthm,amssymb}
\usepackage[capitalize,noabbrev]{cleveref}
\usepackage[square,numbers]{natbib}
\usepackage{nicefrac}
\usepackage{tikz}
\usetikzlibrary{tikzmark}
\usetikzlibrary{backgrounds}
\usetikzlibrary{positioning}
\usetikzlibrary{decorations.pathreplacing,calc}
\usepackage{enumitem}
\usepackage{authblk}
\usepackage{booktabs} % For formal tables

\newtheorem{theorem}{Theorem}
\newtheorem{corollary}{Corollary}
\newtheorem{lemma}[theorem]{Lemma}
\newtheorem{proposition}[theorem]{Proposition}

\newtheorem{innerdefinition}{Definition}
\newtheorem*{innerdefinitionnonumber}{\definitionname}
\providecommand{\definitionname}{}

\usepackage{xparse}
\makeatletter
\NewDocumentEnvironment{definition}{o}
 {  \IfNoValueTF{#1}
    {\innerdefinition}    {     \renewcommand{\definitionname}{#1}     \def\@currentlabel{#1}     \innerdefinitionnonumber    } }
 {\IfNoValueTF{#1}{\endinnerdefinition}{\endinnerdefinitionnonumber}}
\makeatother

\theoremstyle{definition}
\newtheorem{example}{Example}

\title{Sequential Elimination and Union Shapley Value for Group Assessment in Coalitional Games}

\author{Piotr Kępczyński}
\author{Oskar Skibski}
\affil{Institute of Informatics, University of Warsaw}

\begin{document}

\maketitle

\begin{abstract}
Two straightforward methods to extend an assessment of individual elements to groups are to sum individual assessments or to treat the group as a single merged element and assess it accordingly.
In this work, we analyze another natural approach based on sequential elimination: elements of the group are removed one by one, and their assessments are aggregated.
We study this approach in the context of coalitional games and show that, for almost all semivalues, it does not depend on the order of players.
In particular, we introduce a new group value, called the \emph{Union Shapley Value}, and investigate its axiomatic properties.

Our results build on a comprehensive analysis of group values in coalitional games.
Specifically, we define a class of \emph{group (weak consistent) semivalues}---a variant of semivalues satisfying a weak form of monotonicity.
This framework allows us to clarify the differences between existing notions in the literature.
We show that existing group values either assess the total worth of a group or measure its synergy.
We distinguish these two approaches axiomatically and uncover a connection between the corresponding values.
In particular, we show that the well-known Interaction Index \cite{Grabisch:Roubens:1999} is a synergistic counterpart of the value introduced by \citet{Marichal:etal:2007}, which corresponds to the merge approach.
The analysis also yields new synergistic group values associated with the Union Shapley Value, which we call the \emph{Intersection Shapley Value}.
Our results demonstrate that the sequential extension---and the Union Shapley value in particular---constitute one of the most natural extensions of player values to groups in coalitional games.
\end{abstract}

\section{Introduction}\label{section:introduction}

The topic of element assessment appears in a wide variety of areas, including cooperative games~\cite{Shapley:Shubik:1954}, centrality analysis in networks~\cite{Ballester:etal:2006}, and data valuation in machine learning~\cite{Ghorbani:Zou:2019}.
In many scenarios, it becomes essential to assess not only individual elements, but also pairs or groups of elements.
Elements in a group may be independent, but they may also be highly correlated.
As in the case of individual elements, the assessment of groups should not only focus on group performance, but also on the interaction and interplay with the rest of the system.
This is, for example, evident when large machine learning models are considered, where a few features are unable to provide a reasonable prediction on their own~\cite{Lundberg:Lee:2017}.
The central question of this paper is: how should the value of a group of players be defined when individual values are known?

Two straightforward methods to extend an assessment of individual elements to groups are to sum individual assessments or to treat the group as a single merged element and assess it accordingly.
Clearly, taking the sum of individual assessments ignores the underlying structure of cooperation.
In turn, merging alters the analyzed structure and ignores interactions within the group.
As a result, both approaches often fail to accurately capture dependencies within groups.

In this work, we analyze another natural approach based on sequential elimination: elements of a group are removed one by one, and their assessments are aggregated.
Sequential elimination naturally measures the net contribution of a group by discounting overlapping effects and focusing on the damage caused by the group's absence.
This approach is used, for example, as a heuristic method for selecting an optimal group, where the best element is chosen at each step of the process~\cite{Darst:etal:2018,Albert:etal:2000b,Holme:etal:2002}.

We study this approach in coalitional games.
Here, several group values have been proposed in the literature, and the topic has recently received attention in the machine learning literature~\cite{Muschalik:etal:2024,Harris:etal:2022,Chen:etal:2024}.
Arguably, the most important ones are the Interaction Index~\cite{Grabisch:Roubens:1999} and an extension obtained by merging the group into a single element, which we call the Merge Shapley value; it has been proposed in several papers under different names~\cite{Marichal:etal:2007,Flores:etal:2014,Chen:etal:2024}.
However, the outcomes of these extensions may be considered undesirable, as we now demonstrate.

\begin{example}\label{example:intro}
Consider a game in which two players, $A$ and $B$, when they cooperate, create one unit of a good, while two other players, $C$ and $D$, create another unit. No other cooperation exists.
Which pair, $\{A,B\}$ or $\{A,C\}$, is more important?

According to the Shapley value, one of the most important solution concepts, all players are equally important; thus, each pair has the same sum of individual Shapley values.
The same result is obtained by merging the groups: the merged coalition $\{A,B\}$ creates one unit of the good, while the merged coalition $\{A,C\}$ contributes one half to each of the two units.
As a result, both approaches ignore the different roles of the two pairs in the game.
On the other hand, the Interaction Index assigns value zero to the coalition $\{A,C\}$, which does not fairly represent their joint importance, but rather the lack of synergy between them.

From a systemic perspective, removing $\{A,C\}$ eliminates two independent sources of value, whereas removing $\{A,B\}$ removes only one; this asymmetry is precisely captured by the sequential elimination approach.
Moreover, although players $A$ and $C$ are independent, players $A$ and $B$ are not; consequently, the overlap in their contributions should reduce their combined value.
See \Cref{example:12_34} for a formal description.
\hfill $\lrcorner$
\end{example}

Our results demonstrate that extension through sequential elimination constitutes one of the most theoretically grounded extensions of player values to groups in coalitional games.
We show that any consistent semivalue (i.e., a semivalue satisfying the Null Player Out axiom) can be extended to groups using this approach, since sequential elimination does not depend on the order of players.
Moreover, group values obtained through sequential elimination have strong axiomatic properties: they are uniquely characterized by the group extensions of the Balanced Contributions property~\cite{Myerson:1980} and the Potential axiom~\cite{Hart:Mas-Colell:1989}.

Building upon this approach, we propose the \emph{Union Shapley value} as the group extension of the Shapley value.
The Union Shapley value, formally defined as the sum of the dividend shares of coalitions containing at least one member of the group, can also be interpreted as the effect that the removal of this group has on the game.
Hence, among groups of a fixed size, the group with the highest value upon removal damages the game the most.
To achieve this, each coalition's contribution is counted only once, even if it contains multiple members of the group.

The proposed approach is accompanied by a comprehensive analysis of group values.
Specifically, we identify two variants of group values: a group value either assesses the total worth of a group or measures its synergy.
To illustrate the difference, consider an additive game in which the value of every coalition is the sum of the values of singleton coalitions.
If a group value assesses total worth, then the group assessment should be the sum of individual assessments.
In turn, if a group value measures synergy, then clearly no synergy occurs, and each group of size greater than one should receive zero.

We differentiate between both variants through axiomatic characterizations.
Specifically, we first characterize the class of group semivalues using classic axioms, but with a weaker version of monotonicity which, as we argue, is more natural in the context of groups.
Group semivalues include values obtained by summing, merging, and sequential elimination of the Shapley value.
Then, we show that replacing the axiom of Dummy Player with Dummifying Player, introduced by~\citet{Grabisch:Roubens:1999}, leads to the corresponding class of synergy semivalues.
This class includes, for example, the well-known Interaction Index.

Through the functional form of group values, we establish a correspondence between group semivalues and synergy semivalues.
In particular, we show that the Merge Shapley value corresponds to the Interaction Index.
Moreover, we define a synergy semivalue corresponding to the Union Shapley value under the name \emph{Intersection Shapley value}.
We show that both values can serve as natural measures of interaction between pairs of players: when a group of size two is considered, they both sum to the sum of the players' Shapley values; hence, the Union Shapley value measures the \emph{union} of their contributions, while the Intersection Shapley value measures their \emph{intersection}.

\cref{table:summary} summarizes our main axiomatic results.

\textbf{Related work.}
The first definition of a group value was the Interaction Index, proposed by \citet{Grabisch:Roubens:1999}.
The Interaction Index takes into account the contribution of a group to all coalitions, but focuses not on their value, but on their synergy.
See the discussion in \cref{section:synergies} for details.
A variant of this index, under the name Shapley-Taylor Interaction Index, was also proposed by \citet{sundararajan:etal:2020}.

A group value---that we will call the Merge Shapley value---was first proposed by \citet{Marichal:etal:2007} under the name generalized Shapley value.
The idea here is that the value of a group is equal to the value of a merged player in a game where the entire group is treated as a single player.
The authors also proposed a definition of a class of semivalues that violates the null player property; that is, adding a null player to a group may change its value.
We will discuss this in \cref{section:semivalues}.
The Merge Shapley value has also been proposed in several other papers under various names, such as GShap \cite{Chen:etal:2024} and Shapley group value \cite{Flores:etal:2014}, and has been applied to interpret risk detection models~\cite{lin:etal:2022} and to identify critical paths in neural networks~\cite{khakzar:etal:2021}.

Several other extensions of the Shapley value have also been proposed.
\citet{Harris:etal:2022} and \citet{sundararajan:etal:2020} introduced extensions based on an additional parameter that specifies the \emph{order of explanation} defined as the maximal size of considered groups.
\citet{Alshebli:etal:2019} proposed a group value that is not consistent with the Shapley value on singleton coalitions.
In a recent work, \citet{Muschalik:etal:2024} created a library that computes several group values in an application-agnostic framework.

\begin{table}[t]
\begin{tabular}{lcccccc}
\toprule
& \multicolumn{3}{c}{group semivalues} & \multicolumn{3}{c}{synergy semivalues}\\
\toprule
& Union & sum of & Merge & Intersection & Intersection & Interaction \\
& SV & SVs & SV & SV & SV $\cdot |S|$ & Index \\
\midrule 
Linearity & \checkmark & \checkmark & \checkmark & \checkmark & \checkmark & (\checkmark) \\ 
Symmetry & \checkmark & \checkmark & \checkmark & \checkmark & \checkmark & (\checkmark) \\ 
Weak Monotonicity & \checkmark & \checkmark & \checkmark & \checkmark & \checkmark & (\checkmark) \\
Null Player Out & \checkmark & \checkmark & \checkmark & \checkmark & \checkmark & (\checkmark) \\
Dummy Player & \checkmark & \checkmark & \checkmark & & & \\ 
Dummifying Player & & & & \checkmark & \checkmark & (\checkmark) \\ 
\midrule 
Singleton-Efficiency & \checkmark & \checkmark & \checkmark & \checkmark & \checkmark & (\checkmark) \\ 
Group Equality & \checkmark & & & \checkmark & \\
Group Proportionality & & \checkmark & & & \checkmark & \\
\midrule
Dummy Coalition & & & \checkmark & & & \\ 
Potential & (\checkmark) & & & & \\
Balanced Contributions & (\checkmark) & & & & \\
\bottomrule
\end{tabular}\vspace{0.1cm}
\caption{Summary of the main axiomatic characterizations. The symbol \checkmark\ denotes axioms used in the axiomatization of the value, while (\checkmark) means that the axiom is satisfied but not used in the axiomatization.
Group semivalues are defined using Linearity, Symmetry, Weak Monotonicity, Null Player Out, and Dummy Player~(\cref{theorem:semivalues}), while synergy semivalues are defined with Dummifying Player instead of Dummy Player~(\cref{theorem:synergistic_semivalues}).
Adding Singleton-Efficiency implies that a group value is an extension of the Shapley value~(\cref{proposition:efficiency_equivalence}). Group Equality uniquely characterizes the Union Shapley value among group semivalues~(\cref{theorem:axioms_union}) and the Intersection Shapley value among synergy semivalues~(\cref{proposition:axioms_intersection}) that extend the Shapley value.
In turn, Group Proportionality uniquely characterizes the sum of Shapley values among group semivalues~(\cref{theorem:axioms_sum}) and the Intersection Shapley value times the size of the coalition among synergy semivalues~(\cref{proposition:axioms_intersection_s}).
Finally, the Merge Shapley value is uniquely defined within group semivalues through Dummy Coalition~(\cref{proposition:merge_axioms}).
An alternative axiomatization of the Union Shapley value is obtained by combining Potential or Balanced Contributions with Singleton-Efficiency~(\cref{proposition:union_potential_axiom}).}
\label{table:summary}
\end{table}

%%%%%%%%%%%%%%%%%%%%%%%%%%%%%%%%%%%%%%%%%%%%%%
%%%%%%%%%%%%%%%%%%% PRELIMINARIES %%%%%%%%%%%%
%%%%%%%%%%%%%%%%%%%%%%%%%%%%%%%%%%%%%%%%%%%%%%
\section{Preliminaries}\label{section:preliminaries}
A \emph{game} is a pair $(N, v)$ such that $N$ is the set of $n$ players and $v:2^N\rightarrow\mathbb{R}$ is a characteristic function that assigns every coalition some value $v(S)$ (assuming $v(\emptyset)=0$), called the worth of coalition $S$. 
We will denote by $\mathcal{G}$ the set of all games and by $[k]$ the set $\{1,\dots,k\}$ for $k \in \mathbb{N}$.
A game is \emph{monotone} if the value of a larger coalition is not smaller than the value of a smaller one: for every $S \subseteq T \subseteq N$ it holds $v(S) \le v(T)$.

A special class of games is formed by \emph{unanimity games}. For a coalition $S \subseteq N$, a unanimity game $(N, u_S)$ is a game in which all players in $S$ are required for a coalition to obtain the non-zero value of~$1$:
\[
u_S(T)=
\begin{cases}
    1 & \text{if } S\subseteq T, \\
    0 & \text{otherwise.}
\end{cases}
\]
Unanimity games form a basis of all games. Specifically, for every game $(N, v)$ it holds that
\[ 
v=\sum_{S\subseteq N, S\not=\emptyset}\Delta_v(S)\cdot u_S, \quad \text{ where } \quad \Delta_v(S)=\sum_{T\subseteq S}(-1)^{|S|-|T|} v(T).
\] 
The weights $\Delta_v(S)$ are called the \emph{(Harsanyi) dividends} \cite{Harsanyi:1963}.
The notion of dividends will play a central role in our paper.
A dividend represents the additional value created by a coalition that exceeds what can be obtained by its subcoalitions.
This follows from their recursive definition:
\[ \Delta_v(S) = v(S) - \sum_{T \subsetneq S} \Delta_v(T) \quad \text{ for every } \quad S \subseteq N, |S| \ge 2, \]
with the boundary condition: $\Delta_v(\{i\}) = v(\{i\})$ for every $i \in N$.
Importantly, the sum of dividends of all coalitions is equal to the worth of the grand coalition: $\sum_{T \subseteq N} \Delta_v(T) = v(N)$.
If all dividends are non-negative, a game $(N,v)$ is called \emph{positive}.
For example, game $(\{1,2\}, v)$ with $v(\{1\})=v(\{2\})=v(\{1,2\})=1$ is monotone, but it is not positive.

A player $i \in N$ is a \emph{null player} if she does not contribute to any coalition, i.e., $v(S \cup \{i\}) = v(S)$ for every coalition $S \subseteq N \setminus \{i\}$. 
A player $i$ is a \emph{dummy player} if she always contributes the same value: $v(S \cup \{i\}) = v(S) + v(\{i\})$ for every coalition $S \subseteq N \setminus \{i\}$.
It is known that if a coalition contains a dummy player, its dividend is zero, unless it is a singleton coalition.

A \emph{value} of a game is a function $\varphi$ that for every game $(N, v)$ and every player $i\in N$ assigns some real number, denoted by $\varphi_i(N, v)$. 
We will sometimes refer to values as \emph{player values} to distinguish them from group values which are the main focus of this work.

Arguably, one of the most important solution concepts is \emph{the Shapley value}~\cite{Shapley:1953}. 
Let $\Pi(N)$ be the set of all permutations of $N$ and $P^\pi_i=\{j\in N\ |\ \pi(j)<\pi(i)\}$ be the set of predecessors for $i\in N$ and permutation $\pi\in \Pi(N)$.
For a game $(N,v)$, the Shapley value of a player $i$ is defined as follows:
\[
SV_i(N, v)=\frac{1}{|N|!} \sum_{\pi \in \Pi(N)} (v(P^\pi_i \cup \{i\})-v(P^\pi_i)).
\]
Equivalently, the Shapley value is the sum of player's contributions to all coalitions or player's allocated shares of the dividends of all coalitions she belongs to:
\begin{equation*}
SV_i(N, v) = \sum_{T\subseteq N, i\in T}\frac{(|T|-1)!(|N|-|T|)!}{|N|!}(v(T)-v(T\setminus\{i\})) = \sum_{T\subseteq N, i\in T}\dfrac{\Delta_v(T)}{|T|}.
\end{equation*}
The Shapley value is the unique solution that satisfies Efficiency, Symmetry, Additivity and Null-Player.
Axioms pertaining to player values are denoted by a superscript $^\circ$.

\begin{definition}[Efficiency{$^\circ$}]
A (player) value $\varphi$ satisfies \emph{Efficiency} if 
$\sum_{i \in N} \varphi_i(N, v) = v(N)$ for every game $(N,v)$.
\end{definition}
\begin{definition}[Symmetry{$^\circ$}]
A (player) value $\varphi$ satisfies \emph{Symmetry} if $\varphi_i(N, v) = \varphi_{\pi(i)}(N, \pi(v))$ for every game $(N, v)$, every bijection $\pi: N\rightarrow N$ and player $i \in N$.
\end{definition}
\begin{definition}[Additivity{$^\circ$}]
A (player) value $\varphi$ satisfies \emph{Additivity} if 
$\varphi_i(N, v+w)=\varphi_i(N, v)+\varphi_i(N, w)$ for every games $(N,v), (N,w)$ and player $i \in N$.
\end{definition}
\begin{definition}[Null Player{$^\circ$}]
A (player) value $\varphi$ satisfies \emph{Null Player} if $\varphi_i(N,v) = 0$ for every game $(N, v)$ and null player $i\in N$.
\end{definition}
Here, we used the following additional notation.
For arbitrary games $(N,v)$, $(N,w)$, a scalar $c \in \mathbb{R}$ and a bijection $\pi: N \rightarrow N$ we define games $(N,v+w)$, $(N, c\cdot v)$ and $(N, \pi(v))$ as follows: $(v+w)(T) = v(T) + w(T)$, $(c \cdot v)(T) = c \cdot v(T)$ and $(\pi(v))(\pi(T)))=v(T)$ for every $T \subseteq N$, where $\pi(S) = \{\pi(i) : i\in S\}$.

\citet{Myerson:1980} showed that it is possible to replace Symmetry, Additivity and Null Player by one axiom of Balanced Contributions:
\begin{definition}[Balanced Contributions{$^\circ$}]
A (player) value $\varphi$ satisfies \emph{Balanced Contributions} if $\varphi_i(N,v)-\varphi_i(N \setminus \{j\},v) = \varphi_j(N,v) - \varphi_j(N \setminus \{i\}, v)$ for every game $(N, v)$ and two players $i,j \in N$.
\end{definition}

The Shapley value belongs to the class of \emph{semivalues}. 
A value is a semivalue if for every set of players $N$ there exist weights $\beta: [|N|] \rightarrow \mathbb{R}$ satisfying $\sum_{i=1}^n{n \choose i-1}\beta(i)=1$ such that:
\begin{equation}\label{equation:semivalues}
\varphi^{\beta}_i(N, v)=\sum_{T \subseteq N, i \in T} \beta(|T|) (v(T)-v(T\setminus\{i\})).
\end{equation}
\citet{Weber:1988} proved that a value is a semivalue if and only if it satisfies Linearity (a strengthening of Additivity), Symmetry, Dummy Player (a strengthening of Null Player) and Monotonicity:
\begin{definition}[Linearity{$^\circ$}]
A (player) value $\varphi$ satisfies \emph{Linearity} if it satisfies Additivity and $\varphi_i(N, c \cdot v)=c \cdot \varphi_i(N, v)$ for every game $(N,v)$, scalar $c \in \mathbb{R}$ and player $i \in N$.
\end{definition}
\begin{definition}[Dummy Player{$^\circ$}]
A (player) value $\varphi$ satisfies \emph{Dummy Player} if $\varphi_i(N,v) = v(\{i\})$ for every game $(N, v)$ and dummy player $i\in N$.
\end{definition}
\begin{definition}[Monotonicity{$^\circ$}]
A (player) value $\varphi$ satisfies \emph{Monotonicity} if $\varphi_i(N,v) \ge 0$ holds for every monotone game $(N,v)$ and player $i \in N$.
\end{definition}
Another known value that belongs to the class of semivalues is the Banzhaf value, defined through weights $\beta(t) = 1/2^{n-1}$.

According to the definition of semivalues from \cref{equation:semivalues}, weights can depend on the set of players of the game. 
In particular, a value defined as the Shapley value for odd number of players and the Banzhaf value for even number of players is a valid semivalue.
To enforce the consistency between weights, the axiom of Null Player Out is often considered \cite{Derks:Haller:1999}.
\begin{definition}[Null Player Out$^\circ$]
A player value $\varphi$ satisfies \emph{Null Player Out} if and only if $\varphi_i(N,v)=\varphi_i(N\setminus \{j\}, v)$ for every game $(N,v)$, every null player $j\in N$, and every player $i \in N \setminus \{j\}$.
\end{definition}
Null Player Out states that a null player not only receives a zero payoff, but also does not affect the distribution of payoffs among the remaining players.
Specifically, removing a null player from the game does not affect the payoffs of the others.
This axiom is satisfied by all standard values proposed in the literature, including the Shapley value and the Banzhaf value, and it was in fact implied by the original definition of semivalues by \citet{Dubey:etal:1981}.
If a semivalue satisfies Null Player Out, we will call it \emph{consistent}.

\subsection{Group Values}
In this paper, we conduct a systematic study of \emph{group values}.
A \emph{group value} of a game is a function $\varphi$ that for every game $(N, v)$ and every non-empty coalition $S\subseteq N$ assigns some real number, denoted by $\varphi_S(N,v)$.

A group value $\varphi$ \emph{extends} a (player) value $\varphi'$ if $\varphi_{\{i\}}(N,v)=\varphi'_i(N,v)$ for every game $(N, v)$ and player $i\in N$.
A straightforward extension to groups is the sum of individual values.
Such an approach clearly ignores co-dependencies of players. 
That is why different group values have been proposed in the literature.

\citet{Marichal:etal:2007} proposed a group value that we will call the \emph{Merge Shapley value}. For a game $(N,v)$ and coalition $S$ it is defined as follows:
\[
MS_S(N, v)=\sum_{T\subseteq N\setminus S}\dfrac{|T|!(|N|-|T|-|S|)!}{(|N|-|S|+1)!} (v(T\cup S)-v(T)).
\]
According to the Merge Shapley value, the value of a group is equivalent to the Shapley value of a coalition in a game obtained by merging this group into one player. 
Formally, let $(N \setminus S \cup \{[S]\}, v_{[S]})$ denote a game where players from coalition $S$ are merged into player $[S]$: for a game $(N, v)$ and coalitions $S, T\subseteq N$:
\[
v_{[S]}(T)=
\begin{cases}
v(T\setminus\{[S]\} \cup S) & \text{if } [S] \in T, \\
v(T) & \text{otherwise.}
\end{cases}
\]
Now, for a game $(N, v)$, non-empty coalition $S\subseteq N$ it holds:
\[
MS_S(N,v)=SV_{[S]}(N\setminus S\cup\{[S]\}, v_{[S]})
\]

\citet{Marichal:etal:2007} in a similar fashion defined the \emph{generalized semivalues}.
This class corresponds to applying arbitrary semivalues to the game $v_{[S]}$, possibly different for different group sizes.
As a consequence it does not include many natural approaches not based on merging players, such as the sum of individual Shapley values.
We will discuss it in \cref{section:marichal}.

\citet{Grabisch:Roubens:1999} proposed a different group value under the name \emph{Interaction Index}. 
It is defined for game $(N,v)$ and coalition $S$ as follows:
\[
II_S(N, v) = \sum_{T\subseteq N\setminus S}\dfrac{|T|!(|N|-|T|-|S|)!}{(|N|-|S|+1)!}\sum_{R\subseteq S}(-1)^{|S|-|R|} v(T\cup R).
\]
Interaction Index extends the Shapley value with the goal to measure not the worth of a coalition, but the synergy of its members.
As a result, it behaves significantly different than the sum of individual values or the Merge Shapley value.
For example, in an additive game ($v(S) = \sum_{i \in S} v(\{i\})$ for every $S \subseteq N$), both approaches assess each coalition by the sum of singleton values. 
In turn, Interaction Index of every non-singleton coalition is zero. 
Through our analysis we will explain how this concept relates to the two previous ones.

%%%%%%%%%%%%%%%%%%%%%%%%%%%%%%%%%%%%%%%%%%%%%%
%%%%%%%%%%%%%%%%%%% SEQUENTIAL %%%%%%%%%%%%%%%
%%%%%%%%%%%%%%%%%%%%%%%%%%%%%%%%%%%%%%%%%%%%%%
\section{Sequential Extension to Groups}\label{section:union_shapley_value}

So far, we have identified two simple extensions of (player) values to groups.

\begin{definition}[Sum extension]
A group value $\varphi$ is a \emph{sum extension} of a (player) value $\varphi'$ if $\varphi_S(N,v)=\sum_{i\in S}\varphi'_i(N,v)$
for every game $(N,v)$ and every non-empty coalition $S\subseteq N$.
\end{definition}

\begin{definition}[Merge extension]
A group value $\varphi$ is a \emph{merge extension} of a (player) value $\varphi'$ if $\varphi_S(N,v)=\varphi'_{[S]}(N\setminus S\cup\{[S]\},v_{[S]})$ for every game $(N,v)$ and every non-empty coalition $S\subseteq N$.
\end{definition}

Sum extensions simply aggregate individual values, while merge extensions---with the Merge Shapley value as a canonical example---merge players into a single player and then evaluate this merged player.
While the former approach completely ignores dependencies between players, the latter alters the structure of the game and fails to distinguish between independent and highly dependent players.

\begin{example}\label{example:12_34}
Let us formalize the example mentioned in the Introduction.
Fix a game $(N,v)$ with $N = \{1,2,3,4\}$ and $v = u_{\{1,2\}} + u_{\{3,4\}}$.
Consider two coalitions: $\{1,2\}$ and $\{1,3\}$.
According to the Shapley value, all players are equally important, and the sums of Shapley values for both coalitions coincide.
The same conclusion is obtained with the Merge Shapley value:
$MS_{\{1,2\}}(N,v) = MS_{\{1,3\}}(N,v) = 1$.

This is because the game $v$ in which players $\{i,j\}$ are merged into a single player $\oplus$ is equivalent to the game
$(\{\oplus,3,4\},u_{\{\oplus\}}+u_{\{3,4\}})$
for coalition $\{1,2\}$, and to the game
$(\{\oplus,2,4\},u_{\{\oplus,2\}}+u_{\{\oplus,4\}})$
for coalition $\{1,3\}$.
In both cases, the Shapley value of the merged player equals $1$.
\hfill $\lrcorner$
\end{example}

Consider now the idea of sequential elimination: the members of a group are removed one by one, and their values in the resulting games are added.

\begin{definition}[Sequential extension]
A group value $\varphi$ is a \emph{sequential extension} of a (player) value $\varphi'$ if
$\varphi_S(N,v) = \sum_{i\in S}\varphi'_i(N\setminus P^\pi_i, v)$
for every game $(N,v)$, every non-empty coalition $S\subseteq N$, and every permutation $\pi \in \Pi(S)$.
\end{definition}

Note that, for a sequential extension to exist, it is necessary that the order in which players are removed does not affect the resulting sum.
This requirement is equivalent to the Balanced Contributions property, which states that it does not matter whether one first computes the value of player $i$ and then player $j$, or vice versa:
\[
\varphi_i(N,v) + \varphi_j(N \setminus \{i\}, v)
=
\varphi_j(N,v) + \varphi_i(N \setminus \{j\}, v).
\]

\begin{proposition}\label{proposition:sequential_exists}
A (player) value admits a sequential extension if and only if it satisfies Balanced Contributions.
Moreover, this extension is unique.
\end{proposition}

All proofs can be found in the appendix.

As we now show, every consistent semivalue satisfies Balanced Contributions.

\begin{proposition}\label{proposition:semivalues_bc}
Every consistent semivalue satisfies Balanced Contributions.
\end{proposition}

\begin{example}\label{example:12_34_cont1}
Let us continue \cref{example:12_34} and consider the sequential extension $\varphi$ of any player value $\varphi'$ satisfying Balanced Contributions.
We obtain
\begin{align*}
\varphi_{\{1,2\}}(N,v) &= \varphi'_1(N,v) + \varphi'_2(N \setminus \{1\},v), \\
\varphi_{\{1,3\}}(N,v) &= \varphi'_1(N,v) + \varphi'_3(N \setminus \{1\},v).
\end{align*}
Removing player $1$ yields the reduced game
\[
(N \setminus \{1\}, v) = (\{2,3,4\}, u_{\{3,4\}}).
\]
In this game, player $2$ is a null player, while player $3$---as in the original game---jointly generates the worth of $u_{\{3,4\}}$ with player $4$.
Hence, for any reasonable player value $\varphi'$, the group $\{1,2\}$ is assigned a lower value than the group $\{1,3\}$, as contributions of players $1$ and $2$ overlap, while players $1$ and $3$ are independent.
\hfill $\lrcorner$
\end{example}

%%%%%%%%%%%%%%%%%%%%%%%%%%%%%
%%%% UNION SHAPLEY VALUE %%%%
%%%%%%%%%%%%%%%%%%%%%%%%%%%%%
\subsection{Union Shapley value}

To better understand how sequential extensions operate, we now focus on the sequential extension of the Shapley value.
We refer to the unique such extension as the \emph{Union Shapley value}.

The simplest way to define the Union Shapley value is through the notion of a \emph{potential} introduced by \citet{Hart:Mas-Colell:1989}.
A potential of a game is a function that evaluates the game as a whole, rather than individual players or groups.
Consider the following potential function associated with the Shapley value, defined by
\[ P(N,v) = \sum_{T \subseteq N,\, T \neq \emptyset} \frac{\Delta_v(T)}{|T|} = \sum_{T \subseteq N} \frac{(|T|-1)!(|N|-|T|)!}{|N|!}\, v(T). \]
\citet{Hart:Mas-Colell:1989} showed that the Shapley value coincides with the marginal contribution to this function, that is,
\[ SV_i(N,v) = P(N,v) - P(N \setminus \{i\}, v). \]
This implies that, for any order of players $\pi \in \Pi(S)$, it holds that
\[ \sum_{i\in S} SV_i(N\setminus P^\pi_i, v) = P(N,v) - P(N \setminus S,v). \]
From this observation, it follows that the function comparing the potential of the original game with the potential of the game in which the players of a given group have been removed defines a sequential extension.
This construction yields the Union Shapley value.
In other words, the Union Shapley value measures the damage inflicted on the game when a given group leaves.

\begin{definition}[Union Shapley value]
The Union Shapley value of a non-empty coalition $S$ in a game $(N,v)$ is defined by
\begin{equation}\label{equation:union_shapley}
US_S(N,v)
= P(N,v) - P(N \setminus S, v)
= \sum_{T \subseteq N,\, S \cap T \neq \emptyset} \frac{\Delta_v(T)}{|T|}.
\end{equation}
\end{definition}

Simple calculations (see \cref{appendix:calculations} for details) show that the Union Shapley value admits a form close to the classical Shapley value:
\begin{equation*}
US_S(N,v)
= \sum_{T \subseteq N} \frac{(|T|-1)!(|N|-|T|)!}{|N|!}
\bigl(v(T) - v(T \setminus S)\bigr).
\end{equation*}
Similarly to the Shapley value, the Union Shapley value of a group $S$ is expressed as a sum over all coalitions $T \subseteq N$ of marginal contributions $v(T) - v(T \setminus S)$.

The name also admits a natural interpretation in terms of dividends contributing to the value of a group.
Let $\mathcal{S}_i$ denote the set of all coalitions that contain player $i$.
Then,
\[
SV_i(N,v)
= \sum_{T \in \mathcal{S}_i} \frac{\Delta_v(T)}{|T|},
\quad \text{while} \quad
US_S(N,v)
= \sum_{T \in \bigcup_{i \in S} \mathcal{S}_i} \frac{\Delta_v(T)}{|T|}.
\]
Thus, players' contributions to dividends are aggregated, but each coalition contributes at most once.

\begin{example}\label{example:12_34_cont2}
Let us briefly recall \cref{example:12_34}.
The only non-zero dividends are
$\Delta_v(\{1,2\}) = \Delta_v(\{3,4\}) = 1$.
Hence, $US_{\{1,2\}}(N,v) = \frac{\Delta_v(\{1,2\})}{2} = \frac{1}{2}$, and $US_{\{1,3\}}(N,v) = \frac{\Delta_v(\{1,2\})}{2} + \frac{\Delta_v(\{3,4\})}{2} = 1$.
\hfill $\lrcorner$
\end{example}

The characterization of the Union Shapley value via the potential function implies that, among all groups of a given size, the best group is the one that damages the game the most, where the quality of the game is measured by the potential function.
As a consequence, the best group ``covers'' as many distinct parts of the game as possible.
We illustrate this observation with the following two examples.

\begin{example}\label{example:union_positive}
\emph{(Positive Games)}
Consider an arbitrary positive game, that is, a game in which all coalition dividends are non-negative.
Take any payoff vector from the core, a fundamental concept in cooperative game theory consisting of payoff allocations such that every coalition receives at least its value.
Interestingly, the carrier of such a payoff vector (i.e., the set of players receiving non-zero payoffs) forms a group of that size with the highest Union Shapley value.
In fact, the Union Shapley value selects as optimal groups precisely those that contain the carrier of some core payoff vector---whenever such groups exist.
In contrast, the Merge Shapley value and the sum extension of the Shapley value do not exhibit this property.

For a concrete example, consider the game defined by $v = u_{\{1,2\}} + u_{\{3,4,5\}}$, and suppose that the desired group size is~2.
According to the Union Shapley value, groups containing one player from each component are preferred---for example, $\{1,3\}$.
By contrast, both the Merge Shapley value and the sum extension of the Shapley value favor the group $\{1,2\}$, even though the payoff vector $x=(1,1,0,0,0)$ does not lie in the core, whereas $x=(1,0,1,0,0)$ does.

This fact follows from the definition of the potential associated with the Union Shapley value. Players who are not in the carrier of the core form a subgame whose values are zero. Among all positive games, this subgame has the minimal potential. Consequently, the remaining players minimize the potential the most and therefore attain the highest Union Shapley value.
\hfill $\lrcorner$
\end{example}

\begin{example}\label{example:union_induced_subgraph_games}
\emph{(Induced Subgraph Games)}
In this class of games, introduced by \citet{Deng:Papadimitriou:1994}, the game is defined using an undirected graph with weighted edges (for clarity, we assume there are no self-loops). The value of a coalition is given by the sum of the weights of the edges in its induced subgraph. The Shapley value for each player is equal to half of its degree—that is, half the sum of the weights of the edges incident to that player.

In this setting, the Merge Shapley value of a group reduces to the sum of the individual Shapley values, thereby completely ignoring any dependencies between players. In contrast, the Union Shapley value of a group is equal to half the sum of the weights of the edges incident to any node in the group. As a result, the Union Shapley value captures dependencies between nodes: selecting adjacent vertices becomes less rewarding, since they share the contribution of a common edge. In particular, when all edge weights are equal to one, the group of a given size that obtains the highest Union Shapley value corresponds to a vertex cover (i.e., a set of nodes containing at least one endpoint of each edge), if such a set exists.
\hfill $\lrcorner$
\end{example}

%%%%%%%%%%%%%%%%%%%%%%%%%%%%%
%%%% TWO AXIOMATIZATIONS %%%%
%%%%%%%%%%%%%%%%%%%%%%%%%%%%%
\subsection{Two Axiomatizations}

So far, we have introduced the concept of Balanced Contributions to characterize values that admit a sequential extension, and the concept of a potential to define such an extension of the Shapley value.
In this section, we show that both concepts are intrinsically connected to the notion of sequential extension.
Specifically, we establish that the axioms of Balanced Contributions and Potential, generalized to group values, uniquely characterize the sequential extension of any player value.

The first axiomatization follows from our previous discussion.
The axiom of Potential proposed by \citet{Hart:Mas-Colell:1989} formally captures the existence of a potential function for a player value.
It requires the existence of a function such that the value of a player coincides with the marginal contribution to this function.
This axiom naturally extends to groups.

\begin{definition}[Potential]
A group value $\varphi$ satisfies \emph{Potential} if there exists a function $f: \mathcal{G} \rightarrow \mathbb{R}$ with $f(\emptyset,v)=0$ such that $\varphi_S(N,v) = f(N,v) - f(N \setminus S, v)$ for every game $(N,v)$ and every non-empty coalition $S \subseteq N$.
\end{definition}

If a group value satisfies Potential, then it is a sequential extension, as $\varphi_S(N,v)$ for any $\pi \in \Pi(S)$ can be decomposed as 
\[ \varphi_S(N,v) = f(N,v) - f(N \setminus S, v) = \sum_{i \in S} \left(f(N\setminus P^\pi_i, v) - f(N\setminus (P^\pi_i \cup \{i\}), v)\right) = \sum_{i\in S} \varphi'_i(N\setminus P^\pi_i, v). \] 
As we now show, the converse implication also holds: every sequential extension satisfies Potential.
Specifically, the potential of a game is equal to the group value of its grand coalition.

\begin{theorem}\label{theorem:sequential_potential_axiom}
A group value is a sequential extension if and only if it satisfies Potential.
\end{theorem}
\begin{proof}
We already argued that Potential implies that the value is a sequential extension.
Let $\varphi$ be a group value that is a sequential extension of $\varphi'$, and let $f(N,v)=\varphi_N(N,v)$. We will show that $f$ is its potential function. For every game $(N,v)$ and coalition $S\subseteq N$ take any permutation $\pi\in \Pi(N)$ such that $\pi(i)<\pi(j)$ for all $i\in S, j\in N\setminus S$. Then
\begin{multline*}
f(N,v)-f(N\setminus S, v) = \varphi_N(N,v)-\varphi_{N\setminus S}(N\setminus S, v) = \\ 
=\sum_{i\in N}\varphi'_i(N\setminus P^\pi_i,v)-\sum_{i\in N\setminus S}\varphi'_i((N\setminus S) \setminus P^\pi_i,v)
=\sum_{i\in S}\varphi'_i(N\setminus P^\pi_i,v) =\varphi_S(N,v),
\end{multline*}
which concludes the proof.
\end{proof}

The second axiomatization is based on the notion of \emph{Balanced Contributions}.
For individual players, the Balanced Contributions property states that removing one player from the game affects the value of another player in the same way as removing the latter affects the value of the former.
For groups of players, we impose the analogous requirement: groups affect each others' values in the same way.

\begin{definition}[Balanced Contributions]
A group value $\varphi$ satisfies \emph{Balanced Contributions} if, for every game $(N,v)$ and every two coalitions $S, T \subseteq N$, it holds that $\varphi_S(N,v) - \varphi_{S\setminus T}(N \setminus T, v) = \varphi_T(N,v) - \varphi_{T \setminus S}(N \setminus S, v)$, assuming $\varphi_{\emptyset}(N,v) = 0$.
\end{definition}

Fix $T = \{i\} \subseteq S$.
The axiom then imposes that $\varphi_S(N,v) - \varphi_{S \setminus \{i\}}(N \setminus \{i\},v) = \varphi_i(N,v)$, which directly implies the sequential extension.
In the following theorem, we show that every sequential extension satisfies Balanced Contributions as well.

\begin{theorem}\label{theorem:sequential_balanced_axiom}
A group value is a sequential extension if and only if it satisfies Balanced Contributions.
\end{theorem}
\begin{proof}
We already argued that a value satisfying Balanced Contributions is a sequential extension.
From \cref{theorem:sequential_potential_axiom} we know that every value $\varphi$ that is a sequential extension satisfies Potential and its potential function be $f(N,v) = \varphi_N(N,v)$. Now, for every game $(N,v)$ and coalitions $S,T\in N$
\begin{align*}
\varphi_S(N,v)-\varphi_{S\setminus T}(N\setminus T,v) & = f(N,v)-f(N\setminus S,v)-f(N\setminus T,v)+f(N\setminus(S\cup T), v) =\\
& =\varphi_T(N,v)-\varphi_{T\setminus S}(N\setminus S,v),
\end{align*}
which concludes the proof.
\end{proof}

Interestingly, it can be shown that the Potential axiom is connected to the notion of the Union Shapley value, even without requiring the group value to extend the Shapley value.

\begin{proposition}\label{proposition:union_potential_axiom}
If a group value $\varphi$ satisfies Potential, then it satisfies $\varphi_S(N,v) = US_S(N,u_{\varphi})$, where $u_{\varphi}(S) = \sum_{i \in S} \varphi_{\{i\}}(S,v)$ for all $S \subseteq N$.
\end{proposition}

A corollary of \cref{theorem:sequential_balanced_axiom,theorem:sequential_potential_axiom} is an axiomatization of the Union Shapley value.
We use a straightforward analogue of Efficiency, called \emph{Singleton-Efficiency}, which requires that the values of singleton coalitions sum to $v(N)$, just as standard Efficiency does for players.
Note that it is not clear how Efficiency restrictions can be extended to larger groups.

\begin{proposition}\label{proposition:union_potential_efficiency_axiomatization}
The Union Shapley value is the unique group value that satisfies Potential and Singleton-Efficiency.
Moreover, it is the unique group value that satisfies Balanced Contributions and Singleton-Efficiency.
\end{proposition}

Our results demonstrate that the sequential extension---and the Union Shapley value in particular---constitute one of the most natural extensions of player values to groups in coalitional games.

%%%%%%%%%%%%%%%%%%%%%%%%%%%%%%%%%%%%%%%%%%%%%%
%%%%%%%%%%%%%%%%%%% ~SEMIVALUES %%%%%%%%%%%%%%
%%%%%%%%%%%%%%%%%%%%%%%%%%%%%%%%%%%%%%%%%%%%%%

\section{Group semivalues}\label{section:semivalues}

The goal of this section is to study properties of group values.
To this end, we systematically add classic axioms and analyze how they restrict the functional form of group values.
Eventually, we obtain a class of \emph{group weak consistent semivalues}, or simply group semivalues, which is based on a weaker monotonicity condition than that of semivalues, more natural for group values.
Group semivalues encompass the Union Shapley value, as well as other group values studied in the literature, such as the Merge Shapley value and the sum extension of the Shapley value.
The Interaction Index does not belong to this class, but rather to a corresponding class of \emph{synergy semivalues}, which will be the focus of the next section.

We begin by listing the extension to groups of all axioms that we will use in the characterization.

\begin{definition}[Linearity]
A group value $\varphi$ satisfies \emph{Linearity} if 
$\varphi_S(N, v+w)=\varphi_S(N, v)+\varphi_S(N, w)$ and $\varphi_S(N, c\cdot v)=c\cdot \varphi_S(N, v)$ for every games $(N,v), (N,w)$, non-empty coalition $S \subseteq N$ and a scalar $c\in \mathbb{R}$.
\end{definition}
\begin{definition}[Symmetry]
A group value $\varphi$ satisfies \emph{Symmetry} if $\varphi_S(N, v) = \varphi_{\pi(S)}(N, \pi(v))$ for every game $(N, v)$, every bijection $\pi: N\rightarrow N$ and non-empty coalition $S\subseteq N$.
\end{definition}
\begin{definition}[Null Player]
A group value $\varphi$ satisfies \emph{Null Player} if $\varphi_{\{i\}}(N,v) = 0$ and $\varphi_{S\cup \{i\}}(N, v)=\varphi_S(N, v)$ for every game $(N, v)$, null player $i\in N$ and non-empty coalition $S\subseteq N\setminus\{i\}$.
\end{definition}
\begin{definition}[Dummy Player]
A group value $\varphi$ satisfies \emph{Dummy Player} if $\varphi_{\{i\}}(N,v) \!=\! v(\{i\})$ and $\varphi_{S\cup\{i\}}(N, v) = \varphi_S(N, v) + v(\{i\})$ for every game $(N, v)$, dummy player $i\in N$ and non-empty coalition $S\subseteq N\setminus\{i\}$.
\end{definition}
\begin{definition}[Weak Monotonicity]
A group value $\varphi$ satisfies \emph{Weak Monotonicity} if $\varphi_S(N,v) \ge 0$ holds for every positive game $(N,v)$ and non-empty coalition $S\subseteq N$.
\end{definition}
\begin{definition}[Null Player Out]
A group value $\varphi$ satisfies \emph{Null Player Out} if $\varphi_S(N,v) = \varphi_S(N \setminus \{i\}, v)$ for every game $(N, v)$, null player $i\in N$ and non-empty coalition $S\subseteq N\setminus\{i\}$.
\end{definition}

Linearity and Symmetry extend naturally to groups.
Linearity states that the value is a linear function on the set of all games of $N$. 
In turn, Symmetry states that values do not depend on players' names. 
Specifically, if we permute the set of players, we get the permutation of their values and of value of their groups.

The Null Player axiom for player values requires that a null player gets zero value. 
To extend this to groups, it is natural to assume that adding a null player to a group does not affect its value.
Analogically, in a stronger axiom of Dummy Player, we require that adding a dummy player $i$ to the coalition increases its value by $v(\{i\})$.
Both axioms are satisfied by the sum, merge, and sequential extensions of the Shapley value, but they are violated by the Interaction Index---this follows from its synergistic nature that we discuss in \cref{section:synergies}.

Weak Monotonicity is a weaker version of classic Monotonicity and implies that group values are non-negative for positive games.
We argue that for groups Weak Monotonicity is more desirable that the standard Monotonicity axiom that impose non-negativity on monotone games.
To see that, consider a game $(\{1,2\}, v)$ with $v(\{1\}) = v(\{2\}) = v(\{1,2\}) = 1$.
This game is monotone, but the cooperation of players results in a worse outcome.
Hence, the evaluation of coalition $\{1,2\}$ may not necessarily be positive.

Null Player Out was mentioned in the preliminaries as an axiom that is sometimes omitted in the definition of semivalues, but is implied by their original definition.
We include it in our class.
Null Player Out naturally extends to groups: it requires that removing a null player from the game does not affect the payoffs of the remaining groups.

In the following propositions, we show how the consecutive axioms restrict the form of the value.
Linearity implies that the value is a linear combination of the worth of the coalitions.
By using a decomposition of a game to unanimity games, we get that it is also a linear combination of the dividends of the coalition. 

\begin{proposition}\label{theorem:linearity}
A group value $\varphi$ satisfies Linearity if and only if for every set of players $N$ there exists a collection of real weights $(p^S_T)_{S, T\subseteq N}$, such that for every game $(N,v)$ and non-empty coalition $S$ it holds:
\[
\varphi_S(N, v) = \sum_{T\subseteq N, T\not=\emptyset}p^S_T \Delta_v(T).
\]
\end{proposition}

Note that the weights may depend on the set of players; this dependence persists until the Null Player Out axiom is imposed.

As we can see, a group value that satisfies Linearity is a weighted sum of dividends of all coalitions. 
However, when group $S$ is assessed, a weight of a dividend $\Delta_v(T)$, $p^S_T$, may be arbitrary and may depend on players' names.
Adding the Symmetry axiom to Linearity ensures that the weight depends only on (1) the size of $S$, (2) the size of $T$, and (3) the size of their intersection $S \cap T$.
To see why the size of the intersection is also taken into account, note that the dividend of $\Delta_v(\{1\})$ should contribute differently to the value of group $\{1,2\}$ and to the value of group $\{2,3\}$.

\begin{proposition}\label{theorem:symmetry}
A group value $\varphi$ satisfies Linearity and Symmetry if and only if for every set of players $N$ there exists a collection of real weights $(p^{s, q}_t)_{s,t,q\in[n];q\leq s, q\leq t}$, such that for every game $(N, v)$ and non-empty coalition $S$ it holds:
\[
\varphi_S(N, v) = \sum_{T\subseteq N, T\not=\emptyset}p^{|S|, |S\cap T|}_{|T|} \Delta_v(T)
\]
\end{proposition}

Consider adding the Null Player axiom.
As it turns out, if a group value satisfies additionally Null Player, then it cannot depend on dividends of coalitions that are disjoint from the studied coalition, e.g., $\Delta_v(\{1\})$ does not contribute to the value of coalition $\{2,3\}$, as intuition would suggest.
On top of that, the weights of dividends cannot depend directly on the size of the assessed group.

\begin{proposition}\label{theorem:null-player}
A group value $\varphi$ satisfies Linearity, Symmetry and Null Player if and only if for every set of players $N$ there exists a collection of weights $(p^q_t)_{q,t\in [n];q\leq t}$, such that for every game $(N, v)$ and non-empty coalition $S$ it holds:
\[ 
\varphi_S(N, v) = \sum_{T \subseteq N, S \cap T\not=\emptyset}p^{|S\cap T|}_{|T|} \Delta_v(T).
\]
\end{proposition}

Adding the Weak Monotonicity axiom implies that the weights are non-negative and Dummy Player fixes the weight $p^1_1 = 1$.
Finally, adding the Null Player Out axiom makes the weights independent from the set of players in the game.
This leads to the following characterization of semivalues.

\begin{definition}[Group (weak consistent) semivalues]
A group value $\varphi$ is a group (weak consistent) semivalue if there exists an infinite collection of non-negative weights $(p^q_t)_{q,t\in \mathbb{N};q\leq t}$ with $p^1_1=1$, such that for every game $(N, v)$ and non-empty coalition $S$ it holds:
\begin{equation*}
\varphi_S(N, v) = \sum_{T \subseteq N, S \cap T\not=\emptyset}p^{|S\cap T|}_{|T|} \Delta_v(T).
\end{equation*}    
\end{definition}

\begin{theorem}\label{theorem:semivalues}
A group value $\varphi$ satisfies Linearity, Symmetry, Dummy Player, Weak Monotonicity, and Null Player Out if and only if it is a group weak consistent semivalue.
\end{theorem}

Let us now discuss several group values which belong to this class.

\begin{itemize}
\item The Union Shapley value, defined in \cref{section:union_shapley_value}, is a group semivalue with weights $p^q_t = \frac{1}{t}$. Hence, it satisfies all the listed axioms.

\item The Merge Shapley value, defined in \cref{section:preliminaries}, is a group semivalue with weights $p^q_t = \frac{1}{t-q+1}$ (see appendix for the calculations).

\item The sum extension $\varphi$ of the Shapley value is also a group semivalue:
    \[ \varphi_S(N,v) = \sum_{i \in S} SV_i(N,v) = \sum_{T \subseteq N, T \cap S \neq \emptyset} \frac{|S \cap T|}{|T|} \Delta_v(T).\]
    Hence, it is a semivalue with weights $p^q_t = \frac{q}{t}$.

\item The worth of a coalition is also a group semivalue. Specifically, consider $\varphi$ defined as follows:
    \[ \varphi_S(N,v) = v(S) = \sum_{T \subseteq S} \Delta_v(T).\]
    Hence, it is a semivalue with weight $p^q_t = 1$ if $q=t$ (i.e., $|S \cap T| = |T|$) and $p^q_t = 0$, otherwise.
\end{itemize}

Out of listed group values, only the worth of the coalition $v(S)$ is not extending the Shapley value: it extends the value $\varphi$ that assigns each player worth of its singleton coalition: $\varphi_i(N,v) = v(\{i\})$ for every $(N,v), i \in N$.

Group values obtained by sum, merge and sequential extensions can violate many of the axioms considered by us in this section. 
However, as we now show, such extensions exist for all consistent semivalues and belong to group semivalues, satisfying all axioms.

\begin{theorem}\label{theorem:semivalues_player_extensions}
Every player consistent semivalue:
\begin{itemize}
\item has a unique sum extension which is a group weak consistent semivalue;
\item has a unique merge extension which is a group weak consistent semivalue;
\item has a unique sequential extension which is a group weak consistent semivalue.
\end{itemize}
\end{theorem}

These results are derived from the following analysis.

First of all, every player consistent semivalue $\varphi'$ can be represented as a weighted sum over dividends of player's coalitions:
\begin{equation}\label{equation:semivalues_analysis}
\varphi'_i(N,v) = \sum_{T \subseteq N, i \in T} p_{|T|} \Delta_v(T),
\end{equation}
for some collection of non-negative weights $(p_t)_{t \in \mathbb{N}}$.
In particular, the Shapley value is obtained through weights $p_t = 1/t$, while the Banzhaf index through $p_t = 1/2^{t-1}$.
This follows from the reformulation based on the definition of dividends.

A group semivalue $\varphi$ extends the player value $\varphi'$ if weights $p_t^1$ are the same as $p_t$:
\[ p_t^1 = p_t \mbox{ for every }t \in \mathbb{N}. \]
Indeed, for a group semivalue $\varphi$ and $S = \{i\}$ we have:
\begin{equation}\label{equation:semivalues_analysis2} 
\varphi_S(N, v) = \sum_{T \subseteq N, S \cap T\not=\emptyset}p^{|S\cap T|}_{|T|} \Delta_v(T) = \sum_{T \subseteq N, i \in T}p^{1}_{|T|} \Delta_v(T).
\end{equation}
Since dividends are arbitrary, \cref{equation:semivalues_analysis,equation:semivalues_analysis2} coincide if and only if weights are the same.

Recall that for an assessed coalition $S$ and $T$ being the coalition whose dividend is considered, $q$ in the weight $p_t^q$ represents the size of their intersection, i.e., $q = |S \cap T|$, and $t$ represents the size of coalition $T$, i.e., $t = |T|$.  
Now, each extension uniquely characterizes the weights $p_t^q$ based on the weights $p_t^1$:

\begin{table}[t]
\centering
\begin{minipage}{0.32\textwidth}
\centering
\begin{tabular}{c|ccccc}
\multicolumn{6}{c}{Sum extension} \\
\toprule
$t$ \textbackslash{} $q$ & 1 & 2 & 3 & 4 & 5 \\
\midrule
1 & $\nicefrac{1}{1}$\tikzmark{a1l} \\
2 & $\nicefrac{1}{2}$ & \nicefrac{2}{2}\tikzmark{a2l}\\
3 & $\nicefrac{1}{3}$ & \nicefrac{2}{3} & \nicefrac{3}{3}\tikzmark{a3l} \\
4 & $\nicefrac{1}{4}$ & \nicefrac{2}{4} & \nicefrac{3}{4} & \nicefrac{4}{4}\tikzmark{a4l} \\
5 & $\nicefrac{1}{5}$\tikzmark{a1r} & \nicefrac{2}{5}\tikzmark{a2r} & \nicefrac{3}{5}\tikzmark{a3r} & \nicefrac{4}{5}\tikzmark{a4r} & \nicefrac{5}{5} \\
\bottomrule 
\end{tabular}\vspace{0.2cm}
\begin{tikzpicture}[overlay, remember picture]
  \foreach \i in {1,...,4} {
    \draw[line width=0.4pt,gray!30]
      ([xshift=1ex,yshift=2ex]pic cs:a\i l) --
      ([xshift=1ex,yshift=-0.2ex]pic cs:a\i r);
    \node[above right=1px and 2px of pic cs:a\i l,color=gray!75] {\tiny x{\number\numexpr\i+1\relax}};
  }
\end{tikzpicture}
\end{minipage}
\hfill
\begin{minipage}{0.32\textwidth}
\centering
\begin{tabular}{c|ccccc}
\multicolumn{6}{c}{Merge extension} \\
\toprule
$t$ \textbackslash{} $q$ & 1 & 2 & 3 & 4 & 5 \\
\midrule
1 & \tikzmark{m1l}$\nicefrac{1}{1}$ \\
2 & \tikzmark{m2l}$\nicefrac{1}{2}$ & \nicefrac{1}{1}\\
3 & \tikzmark{m3l}$\nicefrac{1}{3}$ & \nicefrac{1}{2} & \nicefrac{1}{1} \\
4 & \tikzmark{m4l}$\nicefrac{1}{4}$ & \nicefrac{1}{3} & \nicefrac{1}{2} & \nicefrac{1}{1} \\
5 & $\nicefrac{1}{5}$ & \tikzmark{m4r}\nicefrac{1}{4} & \tikzmark{m3r}\nicefrac{1}{3} & \tikzmark{m2r}\nicefrac{1}{2} & \tikzmark{m1r}\nicefrac{1}{1} \\
\bottomrule 
\end{tabular}\vspace{0.2cm}
\begin{tikzpicture}[overlay, remember picture]
  \foreach \i in {1,...,4} {
    \draw[line width=0.4pt,gray!30]
      ([yshift=-0ex]pic cs:m\i l) --
      ([yshift=-0ex]pic cs:m\i r);
  }
\end{tikzpicture}
\end{minipage}
\hfill
\begin{minipage}{0.32\textwidth}
\centering
\begin{tabular}{c|ccccc}
\multicolumn{6}{c}{Sequential extension} \\
\toprule
$t$ \textbackslash{} $q$ & 1 & 2 & 3 & 4 & 5 \\
\midrule
1 & $\nicefrac{1}{1}$ \\
2 & \tikzmark{s1l}{$\nicefrac{1}{2}$} & {$\nicefrac{1}{2}$}\tikzmark{s1r} \\
3 & \tikzmark{s2l}{$\nicefrac{1}{3}$} & $\nicefrac{1}{3}$ & {$\nicefrac{1}{3}$}\tikzmark{s2r} \\
4 & \tikzmark{s3l}{$\nicefrac{1}{4}$} & $\nicefrac{1}{4}$ & $\nicefrac{1}{4}$ & {$\nicefrac{1}{4}$}\tikzmark{s3r} \\
5 & \tikzmark{s4l}{$\nicefrac{1}{5}$} & $\nicefrac{1}{5}$ & $\nicefrac{1}{5}$ & $\nicefrac{1}{5}$ & {$\nicefrac{1}{5}$}\tikzmark{s4r} \\
\bottomrule
\end{tabular}\vspace{0.2cm}
\begin{tikzpicture}[overlay, remember picture]
  \foreach \i in {1,...,4} {
    \draw[line width=0.4pt,gray!30]
      ([yshift=2ex]pic cs:s\i l) --
      ([yshift=2ex]pic cs:s\i r);
  }
\end{tikzpicture}
\end{minipage}
\caption{Weights $p_t^q$ of the sum, merge, and sequential extensions of the Shapley value.}
\label{table:weights}
\end{table}

\begin{itemize}

\item A group semivalue is a sum extension if and only if 
$p_t^q = q \cdot p_t^1$ for every $q,t \in \mathbb{N}$ with $q \le t$.  
It is easy to check that the weight $q \cdot p_t^1$ is exactly the weight that appears when player values are added.

\item A group semivalue is a merge extension if and only if 
$p_t^q = p_{t-q+1}^1$ for every $q,t \in \mathbb{N}$ with $q \le t$.  
This means that the only aspect that matters is the number of players in $T$ who do not belong to $S$.  
This is consistent with the merge approach: since the players from $S$ are merged, they get only one share of the dividend of $T$, and $|T \setminus S| + 1$ is the number of players among whom this dividend is shared.

\item A group semivalue is a sequential extension if and only if 
$p_t^q = p_t^1$ for every $q,t \in \mathbb{N}$ with $q \le t$.  
Hence, $q$, the size of the intersection of $S$ and $T$, does not affect the weight.  
This is consistent with the intuition that the sequential extension measures the impact of removing a group from the game: if any of the players from the coalition $T$ is removed, then this dividend ceases to exist.

\end{itemize}

See \cref{table:weights} for an illustration.

%%%%%%%%%%%%%%%%%%%%%%%%%%%%%%%%%
%%%% SEMIVALUES - COMPARISON %%%%
%%%%%%%%%%%%%%%%%%%%%%%%%%%%%%%%%

\subsection{Comparison with generalized semivalues by \citet{Marichal:etal:2007}}\label{section:marichal}

\citet{Marichal:etal:2007} defined a class of \emph{generalized semivalues} as a group value $\varphi$ that for every game $(N,v)$ and $S \subseteq N$ satisfies 
\[ \varphi_S(N,v) = \sum_{T \subseteq N \setminus S} q_{|T|}^{|S|}(n) \left( v(T \cup S) - v(T)\right) \]
for some $\{q_t^s(n)\}_{t=0,\dots,n-s}$ satisfying $\sum_{t=0}^{n-s} \binom{n-s}{t} q_t^s(n) = 1$.

To characterize generalized semivalues, \citet{Marichal:etal:2007} used the \emph{Dummy Coalition} axiom.
Here, a coalition is dummy if added to any coalition of the remaining players always contributes the same.

\begin{definition}[Dummy Coalition by \citet{Marichal:etal:2007}]
A group value $\varphi$ satisfies \emph{Dummy Coalition by \citet{Marichal:etal:2007}} if $\varphi_S(N,v) = v(S)$ for every game $(N,v)$ and dummy coalition $S \subseteq N$. 
A~coalition $S$ is a \emph{dummy coalition} if $v(T \cup S) = v(T) + v(S)$ holds for every $T \subseteq N \setminus S$.
\end{definition}

In the presence of the other axioms, Dummy Coalition implies that interactions between the members of the coalition $S$ and the remaining players are not taken into account.
As a result, generalized semivalues depends only on a subset of values of coalitions.
In particular, the value of group $S$ depends only on values of coalitions that either contain all players from $S$ or none of them.
Hence, the sum extension of the Shapley value, as well as the Union Shapley value that depend on all values, are not included in this class.

On the other hand, generalized semivalues include group values that do not satisfy Dummy Player.
One sample value was defined by \citet{Marichal:etal:2007} under the name of \emph{chaining generalized value}: it is obtained through the weights $q_{t}^{s} = s/(s+t)$.

As we now show, the overlap of both group semivalues and generalized semivalues consists of values that can be obtained using the merge extension.

\begin{proposition}\label{proposition:merge_axioms}
Every group weak consistent semivalue that belongs to \citet{Marichal:etal:2007}'s generalized semivalues is a merge extension.
\end{proposition}

\begin{corollary}\label{corollary:merge_axioms}
The Merge Shapley value is the unique value that extends the Shapley value and satisfies Linearity, Symmetry, Dummy Player, Null Player Out, and Dummy Coalition by \citet{Marichal:etal:2007}.
\end{corollary}

\subsection{Axiomatizations of group values}

\cref{corollary:merge_axioms} contains an axiomatization of the Merge Shapley value.
In this section, we present axiomatizations of the remaining two extensions of the Shapley value: the Union Shapley value and the sum extension.

We begin by considering the implications of the Singleton-Efficiency axiom.

\begin{proposition}\label{proposition:efficiency_equivalence}
For a group weak consistent semivalue with weights $(p_t^q)_{q,t \in \mathbb{N}, q \le t}$, the following conditions are equivalent:
\begin{itemize}
\item it satisfies Singleton-Efficiency;
\item it extends the Shapley value; and
\item $p^1_t = 1/t$.
\end{itemize}
\end{proposition}

Consider now a game in which all players are necessary to generate any value whatsoever, i.e., game $u_N$.
Clearly, all players are equally important, but this raises the question of how to evaluate groups.
We consider two variants.

The first variant states that, in such a case, all groups are equally important as well: every player and every group of players is necessary.
If any player or group abstains from cooperation, no value is created.

The second variant states that the importance of a group should be proportional to its size.
For example, if we interpret the value as a reward, then each player should receive an equal share and, analogously, each group should receive the sum of the shares of its members.

\begin{definition}[Group Equality]
A group value $\varphi$ satisfies \emph{Group Equality} if $\varphi_S(N,u_N) = \varphi_T(N,u_N)$ for every game $(N,u_N)$ and all non-empty coalitions $S,T\subseteq N$.
\end{definition}

\begin{definition}[Group Proportionality]
A group value $\varphi$ satisfies \emph{Group Proportionality} if $\varphi_S(N,u_N)/|S| = \varphi_T(N,u_N)/|T|$ for every game $(N,u_N)$ and all non-empty coalitions $S,T\subseteq N$.
\end{definition}

As we now show, Group Equality and Group Proportionality, when combined with Singleton-Efficiency and the axioms characterizing weak consistent semivalues, uniquely characterize the Union Shapley value and the sum extension of the Shapley value, respectively.

\begin{theorem}\label{theorem:axioms_union}
The Union Shapley value is the unique value that satisfies Linearity, Symmetry, Dummy Player, Weak Monotonicity, Null Player Out, Singleton-Efficiency, and Group Equality.
\end{theorem}

\begin{theorem}\label{theorem:axioms_sum}
The sum extension of the Shapley value is the unique value that satisfies Linearity, Symmetry, Dummy Player, Weak Monotonicity, Null Player Out, Singleton-Efficiency, and Group Proportionality.
\end{theorem}

%%%%%%%%%%%%%%%%%%%%%%%%%%%%%%%%%%%%%%%%%%%%%%
%%%%%%%%%%%%%%%%%%% SYNERGISTIC %%%%%%%%%%%%%%%
%%%%%%%%%%%%%%%%%%%%%%%%%%%%%%%%%%%%%%%%%%%%%%

\section{Class of synergy semivalues}\label{section:synergies}

In this section, we take a different perspective on group values.

To explain the idea, consider an additive game $(N,v)$, i.e., a game in which there is no profit from cooperation and the worth of every coalition is simply the sum of worth of players working alone: $v(S) = \sum_{i \in S} v(\{i\})$ for every $S \subseteq N$. 
The Shapley value of each player is equal to her singleton worth: $SV_i(N,v) = v(\{i\})$ and based on the Dummy Player axiom, satisfied by all group semivalues, the value of every larger group $S$ is equal to its worth: $\varphi_S(N,v) = v(S)$.
However, one can argue that the \emph{synergy} of group $S$ is zero.
In fact, its dividend as well as the Interaction Index are zero.

To capture this, we use Dummifying Player, proposed by \citet{Grabisch:Roubens:1999} (under the name of Dummy Player) that states the value of every non-singleton group that contains a dummy player is zero, and the value of a dummy player is equal to her singleton value.

\begin{definition}[Dummifying Player]
A group value $\varphi$ satisfies \emph{Dummifying Player} if $\varphi_{\{i\}}(N,v) = v(\{i\})$ and $\varphi_{S \cup \{i\}}(N, v) = 0$ for every game $(N, v)$, dummy player $i\in N$ and non-empty coalition $S\subseteq N\setminus\{i\}$.
\end{definition}

As we now show, replacing the Dummy Player axiom with Dummifying Player in the axiomatization of group weak consistent semivalues results in a new class of group values.
We name this class \emph{synergy semivalues}.

\begin{definition}[Synergy semivalues]
A group value $\varphi$ is a \emph{synergy semivalue} if there exists an infinite collection of non-negative weights $(p^q_t)_{q,t\in \mathbb{N};q\leq t}$ with $p^1_1=1$, such that for every game $(N, v)$ and non-empty coalition $S$ it holds:
\begin{equation}\label{equation:synergistic_semivalues_form}
\varphi_S(N, v)=\sum_{S\subseteq T\subseteq N} p^{|S \cap T|}_{|T|}\Delta_v(T) = \sum_{S\subseteq T\subseteq N} p^{|S|}_{|T|}\Delta_v(T).
\end{equation}
\end{definition}

\begin{theorem}\label{theorem:synergistic_semivalues}
A group value $\varphi$ satisfies Linearity, Symmetry, Dummifying Player, Weak Monotonicity and Null Player Out if and only if it is a synergy semivalue.
\end{theorem}

In contrast to group semivalues, synergy semivalues assess a coalition $S$ by a weighted sum of the dividends not of coalitions that overlap with $S$, but of coalitions that contain all players from $S$.
In this way, contributions of individual members are not counted; only synergies among all players are taken into account.
As a result, coalitions that contain any pair of independent players have zero synergy semivalues.

Note that classes of group semivalues and synergy semivalues are defined through the collection of weights of the same type: $(p^q_t)_{q,t\in \mathbb{N},q \leq t}$. 
This fact creates a natural correspondence relation between group values and their synergistic version.
Specifically, we say that a group semivalue \emph{corresponds} to a synergy semivalue if they use the same collection of weights.

As an illustration, consider a group value defined simply as the worth of a coalition: $\varphi_S(N,v) = v(S)$ for every $(N,v), S \subseteq N$.
As we already discussed, it is a group semivalue with weights $p^q_t = 1$ if $q=t$ and $0$, otherwise.
Now, by inserting these weights to \cref{equation:synergistic_semivalues_form}, we get that the synergistic version is simply the dividend of the coalition, i.e., $\varphi$ defined as $\varphi_S(N,v) = \Delta_v(S)$ for every $(N,v), S \subseteq N$.

Worth of a coalition and the dividend of a coalition coincide on singleton coalitions. 
Thus, both group values extends the same player value $\varphi_i(N,v) = v(\{i\})$.
In fact, a group semivalue $\varphi$ and the corresponding synergy semivalue $\hat{\varphi}$ are always two extensions of the same player value, as for every game $(N,v)$ and player $i \in N$ it holds:
\[
\varphi_{\{i\}}(N,v)=\sum_{T\subseteq N, \{i\} \cap T\not=\emptyset} p^{|\{i\}\cap T|}_{|T|} \Delta_v(T) = \sum_{T\subseteq N, i \in T} p^1_{|T|} \Delta_v(T) = \sum_{\{i\}\subseteq T\subseteq N} p^{|\{i\}\cap T|}_{|T|} \Delta_v(T)=\hat{\varphi}_{\{i\}}(N,v).
\]

By recalling that the Union Shapley value is defined through weights $p^q_t = \frac{1}{t}$, we get the definition of a corresponding synergistic group value.

\begin{definition}[Intersection Shapley value]
The Intersection Shapley value of non-empty coalition $S$ in a game $(N,v)$ is defined as follows:
\[ IS_S(N,v) = \sum_{S \subseteq T \subseteq N} \frac{\Delta_v(T)}{|T|} = \sum_{S \subseteq T \subseteq N} \frac{(|T|-1)!(|N|-|T|)!}{|N|!} \sum_{R\subseteq S}(-1)^{|R|}v(T\setminus R)\]
\end{definition}

The Intersection Shapley value has an analogous interpretation to the Union Shapley value: for the Union Shapley value we had:
\[ US_S(N,v) = \sum_{T \in \bigcup_{i \in S} \mathcal{S}_i} \frac{\Delta_v(T)}{|T|}, \quad \text{ and } \quad IS_S(N,v) = \sum_{T \in \bigcap_{i \in S} \mathcal{S}_i} \frac{\Delta_v(T)}{|T|}. \]
Recall that $\mathcal{S}_i = \{S \subseteq N : i \in S\}$.

Let us present corresponding synergistic versions of previously discussed group semivalues.
Interestingly, we show that the Interaction Index corresponds to the Merge Shapley value.

\begin{itemize}

\item The Interaction Index, defined in \cref{section:preliminaries}, is a synergy semivalue with weights $p^q_t = \frac{1}{t-q+1}$. Hence, it is a corresponding synergistic value to the Merge Shapley value.
See \cref{appendix:calculations} for the calculations.

\item The Intersection Shapley value is a synergy semivalue with weights $p^q_t = \frac{1}{t}$ and it is a corresponding synergistic value to the Union Shapley value.

\item The Intersection Shapley value multiplied by the size of the coalition, i.e., $\varphi_S(N,v) = |S| \cdot IS_S(N,v)$ for every $(N,v)$, $S \subseteq N$, is a synergy semivalue with weights $p^q_t = \frac{q}{t}$. Therefore, surprisingly, it is a corresponding synergistic value to the sum extension of the Shapley value. 

\item The dividend of a coalition is a synergy semivalue corresponding to the worth of a coalition, as we already discussed. Note, however, that it does not extend the Shapley value.

\end{itemize}

\begin{example}\label{example:intersection}
Consider again the game from \cref{example:12_34,example:12_34_cont1,example:12_34_cont2}: $v = u_{\{1,2\}} + u_{\{3,4\}}$. 
Based on \cref{equation:synergistic_semivalues_form}, all synergy semivalues assign value $0$ to coalition $T = \{1,3\}$.
For $S = \{1,2\}$, both the Harsanyi dividend and the Interaction Index equal $1$.
In turn, the Intersection Shapley value assigns value $\frac{1}{2}$.
Hence, for both coalitions we get that the sum of the Union Shapley value and the Intersection Shapley value is equal to the sum of the individual Shapley values.
For $S = \{1,2\}$, the sum of Shapley values is $1$, out of which $\frac{1}{2}$ is attributed to their synergy ($IS_S(N,v) = \frac{1}{2}$), so their joint contribution equals $\frac{1}{2}$: $US_S(N,v) = 1-\frac{1}{2} = \frac{1}{2}$. 
In turn, for $T = \{1,3\}$, the sum of Shapley values is $1$, out of which $0$ is attributed to their synergy ($IS_S(N,v) = 0$), so their joint contribution equals $1$: $US_S(N,v) = 1-0=1$.
\hfill $\lrcorner$
\end{example}

The following table summarizes all four pairs of corresponding group values considered in the paper by demonstrating values of groups of different sizes in a symmetric game $u_N$.

\begin{center}
    \setlength{\tabcolsep}{5pt} 
    \begin{tabular}{ccccccccc}
    group semivalue & & synergy semivalue & $1$ & $2$ & $3$ & \ldots & $n-1$ & $n$ \\
    \toprule
    $US_S(v)$ & \small$\longleftrightarrow$ & $IS_S(v)$ & $\nicefrac{1}{n}$ & $\nicefrac{1}{n}$ & $\nicefrac{1}{n}$ & \ldots & $\nicefrac{1}{n}$ & $\nicefrac{1}{n}$ \\
    $MS_S(v)$ & \small$\longleftrightarrow$ & $II_S(v)$ & $\nicefrac{1}{n}$ & $\nicefrac{1}{(n-1)}$ & $\nicefrac{1}{(n-2)}$ & \ldots & $\nicefrac{1}{2}$ & $1$ \\
    $\sum_{i\in S}SV_i(v)$ & \small$\longleftrightarrow$ & $|S| \cdot IS_S(v)$ & $\nicefrac{1}{n}$ & $\nicefrac{2}{n}$ & $\nicefrac{3}{n}$ & \ldots & $\nicefrac{(n-1)}{n}$ & $1$ \\
    $v(S)$ & \small$\longleftrightarrow$ & $\Delta_v(S)$ & $0$  & $0$  & $0$  & \ldots & $0$ & $1$ \\
  \end{tabular}
\end{center}

By making use of the axioms employed to axiomatize the Union Shapley value and the sum extension of the Shapley value, we obtain the following axiomatizations of the Intersection Shapley value and its multiplication, which is a synergy semivalue corresponding to the sum extension.

\begin{proposition}\label{proposition:axioms_intersection}
The Intersection Shapley value is the unique value that satisfies Linearity, Symmetry, Dummifying Player, Weak Monotonicity, Null Player Out, Singleton-Efficiency, and Group Equality.
\end{proposition}

\begin{proposition}\label{proposition:axioms_intersection_s}
The Intersection Shapley value multiplied by the size of the coalition (i.e., $\varphi_S(N,v) = |S| \cdot IS_S(N,v)$ for every game $(N,v)$ and $S \subseteq N$) is the unique value that satisfies Linearity, Symmetry, Dummifying Player, Weak Monotonicity, Null Player Out, Singleton-Efficiency, and Group Proportionality.
\end{proposition}

\subsection{Decomposition of the sum of Shapley values}

The relation identified in \cref{example:intersection} is not a coincidence.
In fact, for every game $(N,v)$ and every two-player coalition $\{i,j\}$, it holds that
\[ US_{\{i,j\}}(N,v) + IS_{\{i,j\}}(N,v) = SV_i(N,v) + SV_j(N,v). \]
Thus, the Union Shapley value and the Intersection Shapley value constitute a decomposition of the sum of Shapley values into a joint value of both players and their shared part, which is accounted for only once in the joint value.
This decomposition can be a useful tool for analyzing interactions between specific players.

As illustrated by \cref{example:intersection}, no other known pair of corresponding group values satisfies this condition.
Moreover, if it is true for any other pair that extends the Shapley value, then they must agree with the Union and Intersection Shapley values on groups of size two.

\begin{proposition}\label{proposition:union_intersection_pair}
If a group weak consistent semivalue $\varphi$ and its corresponding synergy semivalue $\hat{\varphi}$ extend the Shapley value and satisfies $\varphi_{\{i,j\}}(N,v) + \hat{\varphi}_{\{i,j\}}(N,v) = SV_i(N,v) + SV_j(N,v)$ for every game $(N,v)$ and $i,j \in N$, then $\varphi_{\{i,j\}}(N,v) = US_{\{i,j\}}(N,v)$ and $\hat{\varphi}_{\{i,j\}}(N,v) = IS_{\{i,j\}}(N,v)$.
\end{proposition}

For arbitrary coalitions, the relation between the Union Shapley value and the Intersection Shapley value can be described using the inclusion-exclusion principle:
\begin{theorem}\label{theorem:usisrelation}
For every game $(N,v)$ and coalition $S$: $US_S(N,v) \!=\! \sum_{\emptyset \subsetneq T \subseteq S} (-1)^{|T|-1} IS_T(N,v)$.
\end{theorem}
Hence, to compute the joint value of a group, we add the worth of singletons, subtract the overlapping synergies of pairs, add the synergies of triplets (which were subtracted too many times), and so on.

See \cref{table:example_n=3} for an additional example of three-player game and calculations.

\begin{table}[t]
\centering
\begin{minipage}{0.48\textwidth}
\centering
\begin{tabular}{lllll}
$S$ & $v(S)$ & $\sum_{i \in S} SV_i$ & $IS_S$ & $US_S$ \\
\toprule
$\{1\}$ & 1 & 12 & 12 & 12 \\
$\{2\}$ & 3 & 6 & 6 & 6 \\
$\{3\}$ & 4 & 12 & 12 & 12 \\
\midrule
$\{1,2\}$ & 10 & 18 & 3 & 15 \\
$\{1,3\}$ & 21 & 24 & 8 & 16 \\
$\{2,3\}$ & 7 & 18 & 0 & 18 \\
\midrule
$\{1,2,3\}$ & 30 & 30 & 0 & 30\\
\end{tabular}
\end{minipage}
\begin{minipage}{0.48\textwidth}
\centering
\begin{tikzpicture}[
    x=0.7cm,y=0.7cm,
    vertex/.style={circle, draw, minimum size=6mm},
    brace/.style={decorate, decoration={brace, amplitude=10pt, raise=4pt}},
    innerbrace/.style={decorate, decoration={brace, amplitude=4pt, raise=2pt, mirror}}
]

% --- Nodes ---
\node[vertex] (n1) at (1,5) {1};
\node[vertex] (n2) at (5.95,4.29) {2};
\node[vertex] (n3) at (1.71,0.05) {3};

\coordinate (s13) at ($(n1)!0.74!(n3)!4pt!-90:(n3)$);
\coordinate (s31) at ($(n3)!0.74!(n1)!4pt!-90:(n1)$);
\coordinate (b31) at ($(n3)!0.25!(n1)!4pt!-90:(n1)$);

\coordinate (s12) at ($(n1)!0.78!(n2)!4pt!-90:(n2)$);
\coordinate (s21) at ($(n2)!0.37!(n1)!4pt!-90:(n1)$);
\coordinate (b12) at ($(n1)!0.60!(n2)!4pt!-90:(n2)$);

\coordinate (s23) at ($(n2)!0.33!(n3)!4pt!-90:(n3)$);
\coordinate (s32) at ($(n3)!0.67!(n2)!4pt!-90:(n2)$);
\coordinate (b23) at ($(n2)!0.50!(n3)!4pt!-90:(n3)$);

\coordinate (r12) at ($(n1)!0.08!(n2)!4pt!-90:(n2)$);
\coordinate (r21) at ($(n2)!0.08!(n1)!4pt!-90:(n1)$);
\coordinate (r13) at ($(n1)!0.08!(n3)!4pt!-90:(n3)$);
\coordinate (r31) at ($(n3)!0.08!(n1)!4pt!-90:(n1)$);
\coordinate (r23) at ($(n2)!0.07!(n3)!4pt!-90:(n3)$);
\coordinate (r32) at ($(n3)!0.07!(n2)!4pt!-90:(n2)$);

% --- Inner triangle ---
\draw[line width=0.5mm] (n1) -- node[pos=0.4, right, above=1pt] {\footnotesize $\cup$=15} (n2);
\draw[line width=0.5mm] (n1) -- node[pos=0.6, left=19pt, above] {\footnotesize $\cup$=16} (n3);
\draw[line width=0.5mm] (n2) -- node[pos=0.3, right, below=14pt] {\footnotesize $\cup$=18} (n3);

\draw[line width=0.3mm,black!50]
  (r13) --
  node[pos=0.25, left=0pt] {\tiny 12} 
  (s13);

\draw[line width=0.3mm,black!50]
  (r31) --
  node[pos=0.25, right=-2pt] {\tiny 12} 
  (s31);

\draw[line width=0.3mm,black!50]
  (r12) --
  node[pos=0.3, below=-1pt] {\tiny 12} 
  (s12);

\draw[line width=0.3mm,black!50]
  (r21) --
  node[pos=0.4, above=-1pt] {\tiny 6} 
  (s21);

\draw[line width=0.3mm,black!50]
  (r23) --
  node[pos=0.4, left=0pt] {\tiny 6} 
  (s23);

\draw[line width=0.3mm,black!50]
  (r32) --
  node[pos=0.35, right=1pt] {\tiny 12} 
  (s32);

\draw[black!50]
  ($(s23)!0.1cm!90:(r23)$) --
  ($(s23)!0.1cm!-90:(r23)$);

\draw[black!50]
  ($(s32)!0.1cm!90:(r32)$) --
  ($(s32)!0.1cm!-90:(r32)$);

\draw[black!50]
  ($(s13)!0.1cm!90:(r13)$) --
  ($(s13)!0.1cm!-90:(r13)$);
  
\draw[black!50]
  ($(s31)!0.1cm!90:(r31)$) --
  ($(s31)!0.1cm!-90:(r31)$);

\draw[black!50]
  ($(s12)!0.1cm!90:(r12)$) --
  ($(s12)!0.1cm!-90:(r12)$);
  
\draw[black!50]
  ($(s21)!0.1cm!90:(r21)$) --
  ($(s21)!0.1cm!-90:(r21)$);

\draw[line width=0.5mm,blue!50]
  ($(s12)!-0.14cm!90:(r12)$) --
  node[pos=0.6, below=2pt,color=blue!50] {\tiny $\cap$=3}
  ($(s21)!0.14cm!-90:(r21)$);

\draw[line width=0.5mm,blue!50]
  ($(s13)!-0.14cm!90:(r13)$) --
  node[pos=0.5, right=2pt,color=blue!50] {\tiny $\cap$=8}
  ($(s31)!0.14cm!-90:(r31)$);

\end{tikzpicture}

\end{minipage}\vspace{0.1cm}
\caption{The Union Shapley value and the Intersection Shapley value calculations for the game  $v = 1 u_{\{1\}} + 3 u_{\{2\}} + 4 u_{\{3\}} + 6 u_{\{1,2\}} + 16 u_{\{1,3\}}$. 
On the left-hand side, a table of values is shown. 
On the right-hand side, an illustration of the relationships between players is given: 
$\cup$ denotes the Union Shapley value, while $\cap$ denotes the Intersection Shapley value.
As we can see, contributions of players $1$ and $3$ highly overlap, while players $2$ and $3$ are independent.}
\label{table:example_n=3}
\end{table}

\section{Conclusions}

In this paper, we performed a comprehensive analysis of measures for group assessment in coalitional games.
We proposed an approach of sequential elimination in which the elements of the group are removed one by one, and their assessment are added.
We formalized this approach through two axiomatic characterizations and analyzed the extension of the Shapley value, named the Union Shapley value.
Furthermore, we characterized the class of group values and identified a corresponding approach that measures synergy rather than the value of a coalition.
In our analysis we introduced a new measure of synergy and revealed a novel connection between several group values proposed in the literature.

There are several possible directions for future work.
While the new measures have strong theoretical foundation, it is also important to test them empirically in real-world applications of the Shapley value---e.g., to assess interactions between features or evaluate groups of similar features in machine learning algorithms.
Moreover, we plan to study the computational aspects of the new measures and develop dedicated approximation algorithms.
Finally, it would be interesting to extend our definitions to more advanced models of cooperative games, such as graph-restricted games~\cite{Myerson:1977}.

\section*{Acknowledgements}
Piotr Kępczyński and Oskar Skibski were supported by the Polish National Science Centre Grant No. 2023/50/E/ST6/00665.

\newpage

\bibliography{bibliography, piotrek}

\begin{thebibliography}{26}
\providecommand{\natexlab}[1]{#1}
\providecommand{\url}[1]{\texttt{#1}}
\expandafter\ifx\csname urlstyle\endcsname\relax
  \providecommand{\doi}[1]{doi: #1}\else
  \providecommand{\doi}{doi: \begingroup \urlstyle{rm}\Url}\fi

\bibitem[Albert et~al.(2000)Albert, Jeong, and Barab{\'a}si]{Albert:etal:2000b}
R.~Albert, H.~Jeong, and A.-L. Barab{\'a}si.
\newblock Error and attack tolerance of complex networks.
\newblock \emph{Nature}, 406\penalty0 (6794):\penalty0 378--382, 2000.

\bibitem[Alshebli et~al.(2019)Alshebli, Michalak, Skibski, Wooldridge, and
  Rahwan]{Alshebli:etal:2019}
B.~K. Alshebli, T.~P. Michalak, O.~Skibski, M.~Wooldridge, and T.~Rahwan.
\newblock A measure of added value in groups.
\newblock \emph{ACM Transactions on Autonomous and Adaptive Systems},
  13\penalty0 (4):\penalty0 18:1--18:46, 2019.

\bibitem[Ballester et~al.(2006)Ballester, Calv{\'o}-Armengol, and
  Zenou]{Ballester:etal:2006}
C.~Ballester, A.~Calv{\'o}-Armengol, and Y.~Zenou.
\newblock Who's who in networks. wanted: The key player.
\newblock \emph{Econometrica}, 74\penalty0 (5):\penalty0 1403--1417, 2006.

\bibitem[Chen et~al.(2024)Chen, Chen, and Ye]{Chen:etal:2024}
D.~Chen, J.~Chen, and W.~Ye.
\newblock Why groups matter: Necessity of group structures in attributions.
\newblock In \emph{Proceedings of the 5th ACM International Conference on AI in
  Finance}, pages 63--71, 2024.

\bibitem[Darst et~al.(2018)Darst, Malecki, and Engelman]{Darst:etal:2018}
B.~F. Darst, K.~C. Malecki, and C.~D. Engelman.
\newblock Using recursive feature elimination in random forest to account for
  correlated variables in high dimensional data.
\newblock \emph{BMC Genetics}, 19\penalty0 (Suppl 1):\penalty0 65, 2018.

\bibitem[Deng and Papadimitriou(1994)]{Deng:Papadimitriou:1994}
X.~Deng and C.~H. Papadimitriou.
\newblock On the complexity of cooperative solution concepts.
\newblock \emph{Mathematics of Operations Research}, 19\penalty0 (2):\penalty0
  257--266, 1994.

\bibitem[Derks and Haller(1999)]{Derks:Haller:1999}
J.~J. Derks and H.~H. Haller.
\newblock Null players out? linear values for games with variable supports.
\newblock \emph{International Game Theory Review}, 1\penalty0 (03n04):\penalty0
  301--314, 1999.

\bibitem[Dubey et~al.(1981)Dubey, Neyman, and Weber]{Dubey:etal:1981}
P.~Dubey, A.~Neyman, and R.~J. Weber.
\newblock Value theory without efficiency.
\newblock \emph{Mathematics of Operations Research}, 6\penalty0 (1):\penalty0
  122--128, 1981.

\bibitem[Flores et~al.(2014)Flores, Molina, and Tejada]{Flores:etal:2014}
R.~Flores, E.~Molina, and J.~Tejada.
\newblock The shapley group value.
\newblock \emph{arXiv preprint arXiv:1412.5429}, 2014.

\bibitem[Ghorbani and Zou(2019)]{Ghorbani:Zou:2019}
A.~Ghorbani and J.~Zou.
\newblock Data shapley: Equitable valuation of data for machine learning.
\newblock In \emph{International Conference on Machine Learning}, pages
  2242--2251. PMLR, 2019.

\bibitem[Grabisch and Roubens(1999)]{Grabisch:Roubens:1999}
M.~Grabisch and M.~Roubens.
\newblock An axiomatic approach to the concept of interaction among players in
  cooperative games.
\newblock \emph{International Journal of Game Theory}, 28\penalty0
  (4):\penalty0 547--565, 1999.

\bibitem[Harris et~al.(2022)Harris, Pymar, and Rowat]{Harris:etal:2022}
C.~Harris, R.~Pymar, and C.~Rowat.
\newblock Joint shapley values: a measure of joint feature importance.
\newblock In \emph{International Conference on Learning Representations}, 2022.

\bibitem[Harsanyi(1963)]{Harsanyi:1963}
J.~C. Harsanyi.
\newblock A simplified bargaining model for the n-person cooperative game.
\newblock \emph{International Economic Review}, 4\penalty0 (2):\penalty0
  194--220, 1963.

\bibitem[Hart and Mas-Colell(1989)]{Hart:Mas-Colell:1989}
S.~Hart and A.~Mas-Colell.
\newblock Potential, value and consistency.
\newblock \emph{Econometrica}, 57:\penalty0 589--614, 1989.

\bibitem[Holme et~al.(2002)Holme, Kim, Yoon, and Han]{Holme:etal:2002}
P.~Holme, B.~J. Kim, C.~N. Yoon, and S.~K. Han.
\newblock Attack vulnerability of complex networks.
\newblock \emph{Physical review E}, 65\penalty0 (5):\penalty0 056109, 2002.

\bibitem[Khakzar et~al.(2021)Khakzar, Baselizadeh, Khanduja, Rupprecht, Kim,
  and Navab]{khakzar:etal:2021}
A.~Khakzar, S.~Baselizadeh, S.~Khanduja, C.~Rupprecht, S.~T. Kim, and N.~Navab.
\newblock Neural response interpretation through the lens of critical pathways.
\newblock In \emph{Proceedings of the IEEE/CVF conference on computer vision
  and pattern recognition}, pages 13528--13538, 2021.

\bibitem[Lin and Gao(2022)]{lin:etal:2022}
K.~Lin and Y.~Gao.
\newblock Model interpretability of financial fraud detection by group shap.
\newblock \emph{Expert Systems with Applications}, 210:\penalty0 118354, 2022.

\bibitem[Lundberg and Lee(2017)]{Lundberg:Lee:2017}
S.~M. Lundberg and S.-I. Lee.
\newblock A unified approach to interpreting model predictions.
\newblock In \emph{Advances in Neural Information Processing Systems}, pages
  4765--4774, 2017.

\bibitem[Marichal et~al.(2007)Marichal, Kojadinovic, and
  Fujimoto]{Marichal:etal:2007}
J.-L. Marichal, I.~Kojadinovic, and K.~Fujimoto.
\newblock Axiomatic characterizations of generalized values.
\newblock \emph{Discrete Applied Mathematics}, 155\penalty0 (1):\penalty0
  26--43, 2007.

\bibitem[Muschalik et~al.(2024)Muschalik, Baniecki, Fumagalli, Kolpaczki,
  Hammer, and H{\"u}llermeier]{Muschalik:etal:2024}
M.~Muschalik, H.~Baniecki, F.~Fumagalli, P.~Kolpaczki, B.~Hammer, and
  E.~H{\"u}llermeier.
\newblock shapiq: Shapley interactions for machine learning.
\newblock \emph{Advances in Neural Information Processing Systems},
  37:\penalty0 130324--130357, 2024.

\bibitem[Myerson(1977)]{Myerson:1977}
R.~B. Myerson.
\newblock Graphs and cooperation in games.
\newblock \emph{Mathematical Methods of Operations Research}, 2\penalty0
  (3):\penalty0 225--229, 1977.

\bibitem[Myerson(1980)]{Myerson:1980}
R.~B. Myerson.
\newblock Conference structures and fair allocation rules.
\newblock \emph{International Journal of Game Theory}, 9:\penalty0 169--82,
  1980.

\bibitem[Shapley(1953)]{Shapley:1953}
L.~S. Shapley.
\newblock A value for n-person games.
\newblock In H.~Kuhn and A.~Tucker, editors, \emph{Contributions to the Theory
  of Games}, volume~II, pages 307--317. Princeton University Press, 1953.

\bibitem[Shapley and Shubik(1954)]{Shapley:Shubik:1954}
L.~S. Shapley and M.~Shubik.
\newblock A method for evaluating the distribution of power in a committee
  system.
\newblock \emph{American Political Science Review}, 48\penalty0 (03):\penalty0
  787--792, 1954.

\bibitem[Sundararajan et~al.(2020)Sundararajan, Dhamdhere, and
  Agarwal]{sundararajan:etal:2020}
M.~Sundararajan, K.~Dhamdhere, and A.~Agarwal.
\newblock The shapley taylor interaction index.
\newblock In \emph{International conference on machine learning}, pages
  9259--9268. PMLR, 2020.

\bibitem[Weber(1988)]{Weber:1988}
R.~J. Weber.
\newblock Probabilistic values for games.
\newblock \emph{The Shapley Value. Essays in Honor of Lloyd S. Shapley}, pages
  101--119, 1988.

\end{thebibliography}
\bibliographystyle{abbrvnat}

\newpage
\appendix

\section{Alternative Forms of Group Values}\label{appendix:calculations}

In this section, we present the necessary calculations that show the equivalence between two forms of group values: as a weighted sum over values of coalitions and dividends.
We will heavily rely on the following definition of the Harsanyi dividends:
\begin{equation}\label{equation:harsanyi}
v=\sum_{S\subseteq N, S\not=\emptyset}\Delta_v(S)\cdot u_S, \quad \text{ where } \quad \Delta_v(S)=\sum_{T\subseteq S}(-1)^{|S|-|T|} v(T).
\end{equation}
We will also use two following lemmas.

\begin{lemma}\label{lemma:equation:us_transform}
For all sets $Q$ and $N$ such that $Q\subseteq N$ the following holds:
\[ 
\sum_{T: Q\subseteq T \subseteq N} \frac{(|T|-1)!(|N|-|T|)!}{|N|!} = \frac{1}{|Q|}
\]
\end{lemma}
\begin{proof}
We have:
\begin{align*}
\sum_{T: Q\subseteq T \subseteq N}\frac{(|T|-1)!(|N|-|T|)!}{|N|!} = & \sum_{i=0}^{|N|-|Q|}{\binom{|N|-|Q|}{i}}\frac{(|Q|+i-1)!(|N|-|Q|-i)!}{|N|!} \\
= & \sum_{i=0}^{|N|-|Q|}\frac{(|N|-|Q|)!}{i!(|N|-|Q|-i)!}\frac{(|Q|+i-1)!(|N|-|Q|-i)!}{|N|!} \\
= & \frac{(|N|-|Q|)!}{|N|!}\sum_{i=0}^{|N|-|Q|}\frac{(|Q|+i-1)!}{i!} \\
= & \frac{(|N|-|Q|)!}{|N|!} \cdot \frac{|N|!}{|Q|(|N|-|Q|)!} = \frac{1}{|Q|}.
\end{align*}
Here, we used a well-known property of binomial coefficients: $\sum_{i = 0}^b \binom{a+i}{i} = \binom{a+b+1}{b}$ (the left hand side counts the selection of $b$ numbers from $\{1,\dots,a+b+1\}$ where the last element which is not chosen is $a+i+1$); this implies $\sum_{i = 0}^b \frac{(a+i)!}{i!} = \frac{(a+b+1)!}{(a+1) b!}$
\end{proof}

\begin{lemma}\label{lemma:equation:ms_transform}
For all sets $Q$, $S$ and $N$ such that $Q,S\subseteq N$ the following holds:
\[ 
\sum_{T:Q\setminus S\subseteq T\subseteq N\setminus S}\frac{|T|!(|N|-|T|-|S|)!}{(|N|-|S|+1)!}=\frac{1}{|Q|-|S\cap Q|+1}.
\]
\begin{proof}
Fix $e \not \in N$.
By adding $e$ to all sets $T$ we get:
\begin{align*} 
\sum_{T:Q\setminus S\subseteq T\subseteq N\setminus S}\frac{|T|!(|N|-|T|-|S|)!}{(|N|-|S|+1)!} = & \sum_{T: (Q\setminus S \cup \{e\})\subseteq T\subseteq (N\setminus S \cup \{e\})}\frac{(|T|-1)!(|N|-(|T|-1)-|S|)!}{(|N|-|S|+1)!} \\
= & \sum_{T:(Q\setminus S \cup \{e\})\subseteq T\subseteq (N\setminus S \cup \{e\})}\frac{(|T|-1)!(|N\setminus S \cup \{e\}|-|T|)!}{(|N\setminus S \cup \{e\}|)!} \\
= & \frac{1}{|Q \setminus S \cup \{e\}|} = \frac{1}{|Q\setminus S|+1}=\dfrac{1}{|Q|-|S\cap Q|+1}.
\end{align*}
where we used \cref{lemma:equation:us_transform}.

\end{proof}
\end{lemma}

Now, let us consider each group value separately. 

\subsection*{Union Shapley value}
For Union Shapley value, we have:
\[
  US_S(N,v) = \sum_{T \subseteq N} \frac{(|T|-1)!(|N|-|T|)!}{|N|!} (v(T) - v(T \setminus S)).
\]
Using \cref{equation:harsanyi} we get:
\begin{align*}
US_S(N,v) = & \sum_{T:T \subseteq N} \frac{(|T|-1)!(|N|-|T|)!}{|N|!} \left(\sum_{Q:Q\subseteq T}\Delta_v(Q)-\sum_{Q:Q\subseteq T\setminus S}\Delta_v(Q)\right) \nonumber\\
= & \sum_{T:T \subseteq N} \frac{(|T|-1)!(|N|-|T|)!}{|N|!} \left(\sum_{Q:S\cap Q\not=\emptyset, Q\subseteq T}\Delta_v(Q)\right) \nonumber\\ 
= & \sum_{Q:S\cap Q\not=\emptyset,Q\subseteq N}\Delta_v(Q)\sum_{T:Q\subseteq T\subseteq N}\frac{(|T|-1)!(|N|-|T|)!}{|N|!} \nonumber \\
= & \sum_{S\cap Q\not=\emptyset,Q\subseteq N}\frac{\Delta_v(Q)}{|Q|},
\end{align*}
where the last equality comes from \cref{lemma:equation:us_transform}.

\subsection*{Merge Shapley value}
For Merge Shapley value, we have:
\[
  MS_S(N, v)=\sum_{T\subseteq N\setminus S}\dfrac{|T|!(|N|-|T|-|S|)!}{(|N|-|S|+1)!}(v(T\cup S)-v(T)).
\]
Using \cref{equation:harsanyi} we get:
\begin{align*}
MS_S(N, v) = & \sum_{T:T\subseteq N\setminus S}\frac{|T|!(|N|-|T|-|S|)!}{(|N|-|S|+1)!}\left(\sum_{Q:Q\subseteq T\cup S}\Delta_v(Q)-\sum_{Q:Q\subseteq T}\Delta_v(Q)\right) \\
= & \sum_{T:T\subseteq N\setminus S}\frac{|T|!(|N|-|T|-|S|)!}{(|N|-|S|+1)!}\left(\sum_{Q:Q\subseteq T\cup S, Q\cap S\not=\emptyset}\Delta_v(Q)\right) \\
= & \sum_{Q:Q\cap S\not=\emptyset}\Delta_v(Q)\sum_{T:Q\setminus S\subseteq T\subseteq N\setminus S}\frac{|T|!(|N|-|T|-|S|)!}{(|N|-|S|+1)!} = \sum_{Q\cap S\not=\emptyset}\frac{\Delta_v(Q)}{|Q|-|S\cap Q|+1},
\end{align*}
where we used \cref{lemma:equation:ms_transform}.

\subsection*{Intersection Shapley value}
For Intersection Shapley value, we have:
\[
    IS_S(N,v) = \sum_{T:S \subseteq T \subseteq N} \frac{(|T|-1)!(|N|-|T|)!}{|N|!} \sum_{R:R\subseteq S}(-1)^{|R|}v(T\setminus R).
\]
Using \cref{equation:harsanyi} we get:
\begin{align}\label{equation:intersection_shapley}
IS_S(N,v) = & \sum_{T:S \subseteq T \subseteq N} \frac{(|T|-1)!(|N|-|T|)!}{|N|!} \sum_{R:R\subseteq S}(-1)^{|R|}\sum_{Q:Q\subseteq T\setminus R}\Delta_v(Q)\nonumber \\
= & \sum_{T:S \subseteq T \subseteq N} \frac{(|T|-1)!(|N|-|T|)!}{|N|!}\sum_{Q:Q\subseteq T}\Delta_v(Q)\sum_{R:R\subseteq S\setminus Q}(-1)^{|R|} \nonumber\\
= & \sum_{T:S \subseteq T \subseteq N} \frac{(|T|-1)!(|N|-|T|)!}{|N|!}\sum_{Q:S\subseteq Q\subseteq T}\Delta_v(Q)\nonumber \\
= & \sum_{Q:S\subseteq Q\subseteq N}\Delta_v(Q)\sum_{T:Q\subseteq T\subseteq N} \frac{(|T|-1)!(|N|-|T|)!}{|N|!} = \sum_{S\subseteq Q\subseteq N}\frac{\Delta_v(Q)}{|Q|},
\end{align}    
where we used \cref{lemma:equation:us_transform} to get the last equality.

\subsection*{Interaction Index}
For Interaction Index, we have:
\[
II_S(N, v) = \sum_{T:T\subseteq N\setminus S}\frac{|T|!(|N|-|T|-|S|)!}{(|N|-|S|+1)!}\sum_{R:R\subseteq S}(-1)^{|S|-|R|} 
v(T\cup R).
\]
Using \cref{equation:harsanyi} and \cref{lemma:equation:ms_transform} we get:
\begin{align*}
II_S(N, v) = & \sum_{T:T\subseteq N\setminus S}\frac{|T|!(|N|-|T|-|S|)!}{(|N|-|S|+1)!}\sum_{R:R\subseteq S}(-1)^{|S|-|R|} 
\sum_{Q:Q\subseteq T\cup R}\Delta_v(Q) \\
= & \sum_{T:T\subseteq N\setminus S}\frac{|T|!(|N|-|T|-|S|)!}{(|N|-|S|+1)!}\sum_{Q:Q\subseteq T\cup S}\Delta_v(Q) 
\sum_{R:Q\cap S\subseteq R\subseteq S}(-1)^{|S|-|R|} \\
= & \sum_{T:T\subseteq N\setminus S}\frac{|T|!(|N|-|T|-|S|)!}{(|N|-|S|+1)!}\sum_{Q:S\subseteq Q\subseteq T\cup S}\Delta_v(Q) \\
= & \sum_{Q:S\subseteq Q\subseteq N}\Delta_v(Q)\sum_{T:Q\setminus S\subseteq T\subseteq N\setminus S}\frac{|T|!(|N|-|T|-|S|)!}{(|N|-|S|+1)!} = \sum_{S\subseteq Q\subseteq N}\frac{\Delta_v(Q)}{|Q|-|S\cap Q|+1}.
\end{align*}

\section{Omitted Proofs}\label{appendix:proofs}

\subsection*{Proof of \cref{proposition:sequential_exists}}
\begin{proof}
$\Rightarrow:$  
Let $\varphi'$ be a (player) value that admits a sequential extension $\varphi$. We show that $\varphi'$ satisfies Balanced Contributions.  
For every game $(N,v)$ and every pair of players $i,j \in N$ with $i \neq j$, using the formula for the sequential extension, we obtain
\begin{align*}
\varphi'_i(N,v) + \varphi'_j(N\setminus\{i\},v) = \varphi_{\{i,j\}}(N,v) = \varphi'_j(N,v) + \varphi'_i(N\setminus\{j\},v).
\end{align*}

$\Leftarrow:$  
Let $\varphi'$ be a player value that satisfies Balanced Contributions. It is clear that $\varphi'$ admits a sequential extension if and only if, for every game $(N,v)$ and every coalition $S \subseteq N$, the value $\sum_{i \in S} \varphi'_i(N \setminus P_i^{\pi}, v)$ does not depend on the chosen permutation $\pi \in \Pi(S)$.

Consider altering a permutation $\pi \in \Pi(S)$ by swapping the positions of two neighboring players $j_1$ and $j_2$, that is, for $\pi = (\dots, j_1, j_2, \dots)$ we have $\pi' = (\dots, j_2, j_1, \dots)$.
Formally, for $j_1, j_2 \in S$ and $\pi \in \Pi(S)$ such that $\pi(j_1) + 1 = \pi(j_2)$, define $\pi' \in \Pi(S)$ by
\[
\pi'(i)=
\begin{cases}
    \pi(i)+1 & \text{if } i = j_1, \\
    \pi(i)-1 & \text{if } i = j_2, \\
    \pi(i)   & \text{otherwise.}
\end{cases}
\]
Let $T$ denote the set of players that appear before both $j_1$ and $j_2$ in the permutations, that is, $T = P^\pi_{j_1} = P^{\pi'}_{j_2}$.
Using Balanced Contributions, we obtain
\begin{align*}
\sum_{i\in S}\varphi'_i(N\setminus P^{\pi}_i,v)
&=\sum_{i\in S\setminus\{j_1,j_2\}}\varphi'_i(N\setminus P^\pi_i, v)+\varphi'_{j_1}(N\setminus P^\pi_{j_1},v) + \varphi'_{j_2}(N\setminus P^\pi_{j_2},v)\\
&=\sum_{i\in S\setminus\{j_1,j_2\}}\varphi'_i(N\setminus P^\pi_i, v)+\varphi'_{j_1}(N\setminus T,v) + \varphi'_{j_2}(N\setminus (T \cup \{j_1\}),v)\\
&=\sum_{i\in S\setminus\{j_1,j_2\}}\varphi'_i(N\setminus P^\pi_i, v)+\varphi'_{j_2}(N\setminus T,v) + \varphi'_{j_1}(N\setminus (T \cup \{j_2\}),v)\\
&=\sum_{i\in S\setminus\{j_1,j_2\}}\varphi'_i(N\setminus P^{\pi'}_i, v)+\varphi'_{j_2}(N\setminus P^{\pi'}_{j_2},v) + \varphi'_{j_1}(N\setminus P^{\pi'}_{j_1},v)\\
&=\sum_{i\in S}\varphi'_i(N\setminus P^{\pi'}_i,v).
\end{align*}
This shows that exchanging two neighboring players in a permutation does not alter the value.

Since any permutation can be obtained by swapping neighboring players, the value does not depend on the chosen permutation, which proves that $\varphi'$ admits a sequential extension.  
The extension is unique, as it is defined by a direct formula.
\end{proof}

\begin{lemma}\label{lemma:semivalue_npo_divideds_formula}
For every consistent (player) semivalue $\varphi^{\beta}$ there exists some infinite, non-negative collection of weights $(p_t)_{t\in \mathbb{N}}$ with $p_1=1$ such that
\begin{equation}
\varphi_i(N,v)=\sum_{T\subseteq N, i\in T}p_{|T|}\Delta_v(T).
\end{equation}
\end{lemma}
(Please note that unlike with the definition of group semivalues, not every collection of weights satisfying the criteria results in a valid player semivalue.)
\begin{proof}
Fix a set of players $N$.  
Using \cref{equation:harsanyi}, we rewrite the standard formula for semivalues in terms of dividends as follows:
\begin{align*}
\varphi^{\beta}_i(N, v)
&=\sum_{S \subseteq N, i \in S} \beta(|S|) (v(S)-v(S\setminus\{i\}))\\
&=\sum_{S:S \subseteq N, i \in S} \beta(|S|) \sum_{T:T\subseteq S, i\in T}\Delta_v(T)\\
&=\sum_{T:i\in T}\Delta_v(T)  \sum_{S:T\subseteq S}\beta(|S|)\\
&=\sum_{T:i\in T}\sum^{|N|}_{s=|T|}{|N|-|T|\choose{s-|T|}}\beta(s)\Delta_v(T)\\
&=\sum_{T:i\in T}\sum^{|N|-|T|}_{s=0}{|N|-|T|\choose{s}}\beta(s+|T|)\Delta_v(T).
\end{align*}
We can now define a collection of weights $(q^N_t)_{t\in[n]}$ by $q^N_t=\sum^{n-t}_{s=0}{n-t\choose s}\beta(s+t)$, and obtain that
\[
\varphi^{\beta}_i(N,v)=\sum_{i\in T}q^{N}_{|T|}\Delta_v(T)
\]
So far, we have proven that the formula holds for a fixed $N$. 
We need to show that the weights do not depend on $N$. 

To this end, let $N$ be an arbitrary set of players and $j \not \in N$ be an additional player. 
Consider a game $(N\cup\{j\},v)$ in which player $j$ is a null player. 
Since $\varphi^\beta$ satisfies Null Player Out, we get that for every non-empty set of players $T\subseteq N$ and player $i\in T$ it holds:
\begin{align*}
\varphi^\beta_i(N\cup\{j\},u_T) & = \varphi^\beta_i(N,u_T) \\
\sum_{Q\subseteq N\cup\{j\}, i\in Q}q^{N\cup\{j\}}_{|Q|}\Delta_{u_T}(Q) & = \sum_{Q\subseteq N, i\in Q}q^{N}_{|Q|}\Delta_{u_T}(Q) \\
q^{N\cup\{j\}}_{|T|} & = q^N_{|T|}.
\end{align*}
This shows that the weights do not depend on $N$. 
Hence, there exists an infinite set of weights $(p_t)_{t\in \mathbb{N}}$ (defined as $p_t=q^N_t$ for every set of players $N$ such that $|N|\geq t$) such that for every game $(N,v)$ and player $i\in N$ it holds:
\[
\varphi_i(N,v)=\sum_{T\subseteq N, i\in T}p_{|T|}\Delta_v(T).
\]

It remains to prove that the weights are non-negative and that $p_1=1$. 
Non-negativity follows directly from Monotonicity of $\varphi^{\beta}$:
\[
\varphi^\beta_i(N,u_N) = \sum_{T\subseteq N, i\in T}p_{|T|}\Delta_{u_N}(T) = p_t = \geq 0.
\]
In turn, $p_1 = 1$ follows from Dummy Player: in every game $(N,u_{\{i\}})$ player $i\in N$ is a dummy player, hence:
\[
\varphi^\beta_i(N,u_{\{i\}}) = \sum_{T\subseteq N, i\in T} p_{|T|} \Delta_{u_{\{i\}}}(T) = p_1 = 1.
\]
This concludes the proof.
\end{proof}

\subsection*{Proof of \cref{proposition:semivalues_bc}}
\begin{proof}
Let $\varphi$ be an arbitrary consistent semivalue.  
To show that it satisfies Balanced Contributions, we use \cref{lemma:semivalue_npo_divideds_formula}:
\begin{align*}
\varphi_i(N,v)-\varphi_i(N\setminus\{j\},v)
&=\sum_{T\subseteq N, i\in T}p_{|T|}\Delta_v(T)-\sum_{T\subseteq N\setminus\{j\}, i\in T}p_{|T|}\Delta_v(T)
=\sum_{T\subseteq N, i,j\in T}p_{|T|}\Delta_v(T)\\
&=\sum_{T\subseteq N, j\in T}p_{|T|}\Delta_v(T)-\sum_{T\subseteq N\setminus\{i\}, j\in T}p_{|T|}\Delta_v(T)=\varphi_j(N,v)-\varphi_j(N\setminus\{i\},v).
\end{align*}
    
\end{proof}

\subsection*{Proof of \cref{proposition:union_potential_axiom}}

\begin{proof}
Assume that the group value $\varphi$ satisfies Potential, and let $f$ be its potential function.

For $T\subseteq N$, \cref{equation:harsanyi} implies:
\[
\Delta_{u_\varphi}(T)=\sum_{Q\subseteq T}(-1)^{|T|-|Q|}u_\varphi(Q)=\sum_{Q\subseteq T}(-1)^{|T|-|Q|}\sum_{i\in Q}\varphi_{\{i\}}(Q, v).
\]
We can express $\varphi$ in terms of the potential function as follows:
\begin{align*}
\Delta_{u_\varphi}(T) = & \sum_{Q\subseteq T}(-1)^{|T|-|Q|}\sum_{i\in Q} \left(f(Q,v)-f(Q\setminus\{i\},v)\right) \\
= & \sum_{Q\subseteq T}f(Q, v)\left(|Q|(-1)^{|T|-|Q|}-(|T|-|Q|)(-1)^{|T|-(|Q|+1)}\right) \\
= & |T|\sum_{Q\subseteq T}f(Q, v)(-1)^{|T|-|Q|}.
\end{align*}
Combining \cref{equation:union_shapley,equation:harsanyi} with the definitions of the Union Shapley value, we obtain
\begin{align*}
US_S(N, u_\varphi) = & \sum_{T: S\cap T\not= \emptyset}\frac{\Delta_{u_\varphi}(T)}{|T|}=\sum_{T:S\cap T\not= \emptyset}\sum_{Q:Q\subseteq T}f(Q, v)(-1)^{|T|-|Q|} \\
= & \sum_{Q:Q\subseteq N}f(Q, v)\sum_{T:Q\subseteq T, S\cap T\not=\emptyset}(-1)^{|T|-|Q|} \\
= & \sum_{Q:Q\subseteq N}f(Q, v)\left(\sum_{T:Q\subseteq T\subseteq N}(-1)^{|T|-|Q|}-\sum_{T:Q\subseteq T\subseteq N\setminus S}(-1)^{|T|-|Q|}\right).
\end{align*}
Notice that both sums in the parentheses are equal to zero, unless $Q=N$ for the first sum and $Q=N\setminus S$ for the second sum.
Thus, the expression simplifies to
\[
US_S(N, u_\varphi) = f(N, v)-f(N\setminus S, v) = \varphi_S(N, v),
\]
which concludes the proof.
\end{proof}

\subsection*{Proof of \cref{proposition:union_potential_efficiency_axiomatization}}
\begin{proof}
We begin by showing that if a group value $\varphi$ satisfies Singleton-Efficiency and Potential or Balanced Contributions, then it extends the Shapley value.

Let $\varphi$ be a group value extending $\varphi'$. 
Clearly, if $\varphi$ satisfies Singleton-Efficiency, then $\varphi'$ satisfies Efficiency: for every game $(N,v)$, we have
\[
\sum_{i\in N} \varphi'_i(N,v) = \sum_{i\in N} \varphi_{\{i\}}(N,v) = v(N).
\]

Furthermore, we show that the same holds regarding Potential. 
Let $f$ be the potential of $\varphi$. 
For every game $(N,v)$ and every player $i\in N$,
\[
\varphi'_i(N,v)=\varphi_{\{i\}}(N,v)=f(N,v)-f(N\setminus\{i\},v).
\]
Hence, $f$ is also the potential function for $\varphi'$. 

Lastly, we show that the same is true for Balanced Contributions. If $\varphi$ satisfies Balanced Contributions, then
\begin{multline*}
\varphi'_i(N,v)-\varphi'_j(N\setminus\{i\},v)=\varphi_{\{i\}}(N,v)-\varphi_{\{j\}}(N\setminus\{i\},v)=\\
\varphi_{\{j\}}(N,v)-\varphi_{\{i\}}(N\setminus\{j\},v)=\varphi'_j(N,v)-\varphi'_i(N\setminus\{j\},v).
\end{multline*}

We have thus shown that if $\varphi$ satisfies Singleton-Efficiency and Potential, then $\varphi'$ satisfies Efficiency and Potential. 
From \cite{Hart:Mas-Colell:1989}, we know that $\varphi'$ is the Shapley value.
Analogously, if $\varphi$ satisfies Singleton-Efficiency and Balanced Contributions, then $\varphi'$ satisfies Efficiency and Balanced Contributions, and from \citet{Myerson:1977} we know that $\varphi'$ is the Shapley value.

Take any group value satisfying Singleton-Efficiency and Potential (or Singleton-Efficiency and Balanced Contributions). 
Our analysis shows that it extends the Shapley value.
In turn, \cref{theorem:sequential_potential_axiom,theorem:sequential_balanced_axiom} show that it is a sequential extension.
Since \cref{proposition:sequential_exists} establishes the existence of a unique sequential extension of the Shapley value, and the Union Shapley value is such an extension, we conclude that this unique value is the Union Shapley value.
\end{proof}

\subsection*{Proof of \cref{theorem:linearity}}
\begin{proof}
$\Rightarrow:$ Let $\varphi$ be a group value satisfying Linearity. 
Define $p^S_T=\varphi_S(N, u_T)$. 
From Linearity, \cref{equation:harsanyi}, and the definition of the Harsanyi dividends, we obtain
\begin{equation*}
\varphi_S(N, v) = 
\sum_{T\subseteq N, T\not=\emptyset}\Delta_v(T)\cdot\varphi_S(N, u_T)
= \sum_{T\subseteq N, T\not=\emptyset} p^S_T \Delta_v(T).
\end{equation*}

$\Leftarrow:$ Let $\varphi$ be a group value defined through a collection of real weights $(p^S_T)_{S, T\subseteq N}$. 
The linearity of the dividends implies
\begin{align*}
\varphi_S(N,v+w) 
&= \sum_{T\subseteq N, T\not=\emptyset} p^S_T \Delta_{v+w}(T) \\
&= \sum_{T\subseteq N, T\not=\emptyset} p^S_T \Delta_v(T) + \sum_{T\subseteq N, T\not=\emptyset} p^S_T \Delta_w(T) = \varphi_S(N, v) + \varphi_S(N, w).
\end{align*}

Similarly, for every $c\in \mathbb{R}$, using the linearity of the dividends again, we have
\[
\varphi_S(N, c\cdot v) 
= \sum_{T\subseteq N, T\not=\emptyset} p^S_T \Delta_{c\cdot v}(T) 
= c \cdot \sum_{T\subseteq N, T\not=\emptyset} p^S_T \Delta_v(T) 
= c \cdot \varphi_S(N, v).
\]
\end{proof}

\subsection*{Proof of \cref{theorem:symmetry}}
\begin{proof}
$\Rightarrow:$ 
Let $\varphi$ be a group value satisfying Linearity and Symmetry.
From \cref{theorem:linearity}, we know that it is of the form
\[
\varphi_S(N, v) = \sum_{T\subseteq N, T\not=\emptyset} \hat{p}^S_T \, \Delta_v(T).
\]
Moreover, for every $S, T \subseteq N$, we have $\hat{p}^S_T = \varphi_S(N, u_T)$.

We need to show that the weights depend only on the sizes of the sets $S$ and $T$ and on their intersection $S \cap T$. 
To this end, take four sets $S, S', T, T' \subseteq N$ such that $|S| = |S'|$, $|T| = |T'|$, and $|S \cap T| = |S' \cap T'|$.
Clearly, there exists a bijection $\pi : N \to N$ such that $\pi(S) = S'$ and $\pi(u_T) = u_{T'}$. 
By Symmetry, we then have
\[
\hat{p}^S_T = \varphi_S(N, u_T) = \varphi_{\pi(S)}(N, \pi(u_T)) = \varphi_{S'}(N, u_{T'}) = \hat{p}^{S'}_{T'}.
\]

$\Leftarrow:$ 
Let $\varphi$ be a group value defined through a collection of real weights $(p^{s,q}_t)_{s,t,q\in\mathbb{N}; q\le s, q\le t}$.
From \cref{theorem:linearity}, we know that $\varphi$ satisfies Linearity. 
It remains to verify Symmetry.
For every game $(N,v)$, bijection $\pi$, and set $S \subseteq N$, we have
\begin{align*}
\varphi_S(N, v) 
&= \sum_{T\subseteq N, T\not=\emptyset} p^{|S|, |S\cap T|}_{|T|} \, \Delta_v(T) \\
&= \sum_{T\subseteq N, T\not=\emptyset} p^{|\pi(S)|, |\pi(S)\cap \pi(T)|}_{|\pi(T)|} \, \Delta_{\pi(v)}(\pi(T)) \\
&= \varphi_{\pi(S)}(N, \pi(v)),
\end{align*}
which shows that $\varphi$ satisfies Symmetry.
\end{proof}

\subsection*{Proof of \cref{theorem:null-player}}
\begin{proof}
$\Rightarrow:$ 
Let $\varphi$ be a group value satisfying Linearity, Symmetry, and the Null Player axiom.
From \cref{theorem:symmetry}, we know that $\varphi$ is of the form
\[
\varphi_S(N, v) = \sum_{T\subseteq N, T\not=\emptyset} p^{|S|, |S\cap T|}_{|T|} \, \Delta_v(T).
\]
Moreover, for every $S, T \subseteq N$, we have $p^{|S|, |S\cap T|}_{|T|} = \varphi_S(N, u_T)$.

We need to show that the weight $p^{|S|, |S\cap T|}_{|T|}$ is zero if $|S \cap T| = 0$, and that it does not depend on $|S|$.
To this end, consider the game $u_T$, in which all players not in $T$ are null players.

First, let $S, T \subseteq N$ be such that $S \cap T = \emptyset$.
In game $u_T$, all players in $S$ are null players. 
By the Null Player axiom, we then have $p^{|S|,0}_{|T|} = \varphi_S(N, u_T) = 0$. 

Second, let $S, T \subseteq N$ be such that $S \cap T \neq \emptyset$.
Since all players not in $T$ are null players, the Null Player axiom gives
\[
p^{|S|, |S\cap T|}_{|T|} = \varphi_S(N, u_T) = \varphi_{S\cap T}(N, u_T) = p^{|S\cap T|, |S\cap T|}_{|T|},
\]
which shows that the weights do not depend on $|S|$, concluding the proof.

$\Leftarrow:$ 
Let $\varphi$ be a group value defined through a collection of real weights $(p^q_t)_{q,t\in\mathbb{N}, q \le t}$.
From \cref{theorem:symmetry}, we know that $\varphi$ satisfies Linearity and Symmetry.
It remains to verify the Null Player property.
For every game $(N, v)$, null player $i \in N$, and coalition $S \subseteq N$, we have
\begin{align*}
\varphi_S(N, v) 
&= \sum_{S \cap T \neq \emptyset} p^{|S \cap T|}_{|T|} \, \Delta_v(T) \\
&= \sum_{\substack{S \cap T \neq \emptyset \\ i \notin T}} p^{|S \cap T|}_{|T|} \, \Delta_v(T) + \sum_{\substack{S \cap T \neq \emptyset \\ i \in T}} p^{|S \cap T|}_{|T|} \, \Delta_v(T).
\end{align*}
Since dividends of all coalitions containing a null player are zero, the second sum is $0$, giving
\begin{align*}
\varphi_S(N, v) 
&= \sum_{\substack{S \cap T \neq \emptyset \\ i \notin T}} p^{|S \cap T|}_{|T|} \, \Delta_v(T) \\
&= \sum_{\substack{(S \cup \{i\}) \cap T \neq \emptyset \\ i \notin T}} p^{|(S \cup \{i\}) \cap T|}_{|T|} \, \Delta_v(T) \\
&= \sum_{(S \cup \{i\}) \cap T \neq \emptyset} p^{|(S \cup \{i\}) \cap T|}_{|T|} \, \Delta_v(T) - \sum_{\substack{(S \cup \{i\}) \cap T \neq \emptyset \\ i \in T}} p^{|(S \cup \{i\}) \cap T|}_{|T|} \, \Delta_v(T) \\
&= \sum_{(S \cup \{i\}) \cap T \neq \emptyset} p^{|(S \cup \{i\}) \cap T|}_{|T|} \, \Delta_v(T) = \varphi_{S \cup \{i\}}(N, v).
\end{align*}
Here, we used that the dividends of all coalitions containing the null player $i$ are zero.
\end{proof}

\subsection*{Proof of \cref{theorem:semivalues}}
\begin{proof}
$\Rightarrow:$ 
Let $\varphi$ be a group value satisfying Linearity, Symmetry, Dummy Player, Weak Monotonicity, and Null Player Out.
From \cref{theorem:null-player}, it holds that for a fixed $N$, $\varphi$ is of the form
\[
\varphi_S(N, v) = \sum_{T \subseteq N, S \cap T \neq \emptyset} p^{|S \cap T|}_{|T|} \, \Delta_v(T).
\]
Moreover, for every $S, T \subseteq N$, we have $p^{|S \cap T|}_{|T|} = \varphi_S(N, u_T)$.

It remains to show that the weights are independent of $N$, that $p^1_1 = 1$, and that all weights are non-negative.
Let $(p^q_t(N))_{q,t \in [n], q \le t}$ denote the collection of weights of $\varphi$ for a given $N$. 
For every set of players $N$ and player $i \notin N$, consider a game $(N \cup \{i\}, v)$ in which $i$ is a null player. 
Then, for every $q, t \in [n]$ with $q \le t$, take any coalitions $S, Q \subseteq N$ such that $|S \cap Q| = q$ and $|Q| = t$.
From Null Player Out, we obtain
\[ p^q_t(N) = p^{|S \cap Q|}_{|Q|}(N) = \varphi_S(N, u_Q) = \varphi_S(N \cup \{i\}, u_Q) = p^{|S \cap Q|}_{|Q|}(N \cup \{i\}) 
= p^q_t(N \cup \{i\}), \]
which shows that the weights do not change when adding a new player for a game.

Fix another set of players $N'$. Since we have just argued that adding a player does not change the weights, by repeatedly adding players we obtain
\[
p_t^q(N) = p_t^q(N \cup N') = p_t^q(N'),
\]
which shows that the weights do not depend on the set of players.

Next, for every $q, t \in \mathbb{N}$ with $q \le t$, fix sets of players $S, T \subseteq N$ and player $i \in N$ such that $|S \cap T| = q$ and $|T| = t$.  
Since the game $(N, u_T)$ is positive, Weak Monotonicity gives 
\[
\varphi_S(N, u_T) = p^{|S \cap T|}_{|T|} = p^q_t \ge 0.
\] 
Furthermore, in the game $(N, u_{\{i\}})$, player $i$ is a dummy player, so
\[
\varphi_{\{i\}}(N, u_{\{i\}}) = p^{|\{i\}|}_{|\{i\}|} = p^1_1 = 1.
\]

Hence, $\varphi$ is a group weak consistent semivalue.

$\Leftarrow:$ 
From \cref{theorem:null-player}, we know that $\varphi$ satisfies Linearity and Symmetry. 
Let us focus on Dummy Player, Weak Monotonicity and Null Player Out.
Fix $N,i \in N$ and $S \subseteq N \setminus \{i\}$.
It is known that for every game $(N, v)$ and dummy player $i\in N$, $\Delta_v(\{i\})=v(\{i\})$ and for all coalitions $\{i\}\subsetneq S$, $\Delta_v(S)=0$. This means that:
\[
\varphi_{\{i\}}(N, v)=\sum_{\{i\}\cap T\not=\emptyset}p^{|\{i\}\cap T|}_{|T|}\Delta_v(T)=p^1_1\Delta_v(\{i\})=v(\{i\}).
\]
For Dummy Player, it remains to be shown that $\varphi_{S\cup\{i\}}(N, v)=\varphi_S(N, v)+v(\{i\})$. 
We have:
\begin{align*}
\varphi_{S\cup\{i\}}(N, v)
= & \sum_{(S\cup\{i\})\cap T\not= \emptyset}p^{|(S\cup\{i\})\cap T|}_{|T|}\Delta_v(T)\\
= & \sum_{(S\cup\{i\})\cap T\not= \emptyset, i\in T}p^{|(S\cup\{i\})\cap T|}_{|T|}\Delta_v(T)+\sum_{(S\cup\{i\})\cap T\not= \emptyset, i\not\in T}p^{|(S\cup\{i\})\cap T|}_{|T|}\Delta_v(T).
\end{align*}
Since all dividends of coalitions with a dummy player are $0$ except of the dividend of the singleton coalition (equal to $v(\{i\})$), the left-hand sum simplifies to $p^1_1 \cdot v(\{i\}) = v(\{i\})$.
For the right-hand sum, note that $(S \cup \{i\}) \cap T = S \cap T$:
\begin{align*}
\varphi_{S\cup\{i\}}(N, v)
& = p_1^1 \cdot v(\{i\}) + \sum_{S \cap T \neq \emptyset, i \not \in T} p^{|S \cap T|}_{|T|}\Delta_v(T) \\
& = v(\{i\}) + \sum_{S \cap T \neq \emptyset} p^{|S \cap T|}_{|T|}\Delta_v(T) = v(\{i\}) + \varphi_S(N,v).
\end{align*}
Here, we used again the fact that $\Delta_v(T) = 0$ if $i \in T$.
This concludes the proof that $\varphi$ satisfies Dummy Player.

For Weak Monotonicity, we know that all the weights are non-negative.
For every positive game $(N,v)$ we know that the dividends are also non-negative. This means that for every coalition $S\subseteq N$ it holds that:
\[
\varphi_S(N, v)=\sum_{S\cap T\not= \emptyset}p^{|S\cap T|}_{|T|}\Delta_v(T)\geq 0.
\]
This shows that $\varphi$ satisfies also Weak Monotonicity.

We also directly show that $\varphi$ satisfies Null Player Out. For every game $(N,v)$, null player $i\in N$ and non-empty coalition $S$ we have
\[
\varphi_S(N,v)=\sum_{T\subseteq N, S\cap T\not=\emptyset}p^{|S\cap T|}_{|T|}\Delta_v(T)=\sum_{T\subseteq N\setminus\{i\}, S\cap T\not=\emptyset}p^{|S\cap T|}_{|T|}\Delta_v(T)=\varphi_S(N\setminus\{i\},v),
\]
which holds because all dividends for coalitions containing a null player are equal to zero.
This concludes the proof.
\end{proof}

\begin{lemma}\label{lemma:merged_dividends}
For every game $(N,v)$ and coalition $S\subseteq N$, the dividends of the merged game $(N \setminus S \cup \{[S]\}, v_{[S]})$ are equal to
\[
\Delta_{v_{[S]}}(T)=
\begin{cases}
    \sum_{R\subseteq S, R\not=\emptyset}\Delta_v(T\setminus\{[S]\}\cup R) & \text{if } [S]\in T,\\
    \Delta_v(T) & \text{otherwise. }
\end{cases}
\]
\begin{proof}
If $[S]\not\in Q$, then using \cref{equation:harsanyi} we get
\[
\Delta_{v_{[S]}}(Q)=\sum_{T\subseteq Q}(-1)^{|Q|-|T|} v_{[S]}(T)=\sum_{T\subseteq Q}(-1)^{|Q|-|T|} v(T)=\Delta_{v}(T)
\]
If $[S]\in Q$, then
\begin{align*}
\Delta_{v_{[S]}}(Q)
&=\sum_{T\subseteq Q}(-1)^{|Q|-|T|} v_{[S]}(T)\\
&=\sum_{T\subseteq Q\setminus \{[S]\}}(-1)^{|Q|-|T|} v_{[S]}(T)+\sum_{T\subseteq Q\setminus \{[S]\}}(-1)^{|Q|-|T\cup\{[S]\}|} v_{[S]}(T\cup\{[S]\})\\
&=-\sum_{T\subseteq Q\setminus \{[S]\}}(-1)^{|Q\setminus{\{[S]\}}|-|T|} v(T)+\sum_{T\subseteq Q\setminus \{[S]\}}(-1)^{|Q\setminus\{[S]\}|-|T|} v(T\cup S)\\
&=\sum_{T\subseteq Q\setminus \{[S]\}}(-1)^{|Q\setminus{\{[S]\}}|-|T|}\left(v(T\cup S)-v(T) \right)\\
&=\sum_{T\subseteq Q\setminus \{[S]\}}(-1)^{|Q\setminus{\{[S]\}}|-|T|}\left(\sum_{U\subseteq T\cup S}\Delta_v(U)-\sum_{U\subseteq T}\Delta_v(U)\right)\\
&=\sum_{T\subseteq Q\setminus \{[S]\}}(-1)^{|Q\setminus{\{[S]\}}|-|T|}\sum_{U\subseteq T}\sum_{R\subseteq S, R\not=\emptyset}\Delta_v(U\cup R)\\
&=\sum_{R\subseteq S, R\not=\emptyset}\sum_{T\subseteq Q\setminus \{[S]\}}(-1)^{|Q\setminus{\{[S]\}}|-|T|}\sum_{U\subseteq T}\Delta_v(U\cup R)\\
&=\sum_{R\subseteq S, R\not=\emptyset}\sum_{U\subseteq Q\setminus\{[S]\}}\Delta_v(U\cup R)\sum_{T\subseteq Q\setminus \{[S]\}, U\subseteq T}(-1)^{|Q\setminus{\{[S]\}}|-|T|}\\
&=\sum_{R\subseteq S, R\not=\emptyset}\Delta_v((Q\setminus\{[S]\})\cup R).
\end{align*}
Here, we used the fact that 
\[ \sum_{T\subseteq Q\setminus \{[S]\}, U\subseteq T}(-1)^{|Q\setminus{\{[S]\}}|-|T|} = \begin{cases} 1 & \mbox{if } U=Q\setminus\{[S]\},\\0 & \mbox{otherwise.}\end{cases}\] 
This concludes the proof.
\end{proof}

\end{lemma}

\begin{lemma}\label{lemma:gws:merge_weights}
A group weak consistent semivalue is a merge extension if its weights satisfy $p^q_t=p^{q+1}_{t+1}$.
\begin{proof}
From \cref{lemma:merged_dividends} and the weight dependency $p^q_t=p^{q+1}_{t+1}$, we get that $\varphi$ is a merge extension. 
Specifically, for every game $(N,v)$, and coalition $S\subseteq N$ we obtain
\begin{align*}
\varphi_{[S]}(N\setminus & S\cup\{[S]\}, v_{[S]})=
\sum_{T\subseteq N\setminus S\cup\{[S]\}, [S]\in T}p^{|[S]\cap T|}_{|T|}\Delta_{v_{[S]}}(T)\\
&=\sum_{T\subseteq N\setminus S\cup\{[S]\}, [S]\in T}p^1_{|T|}\sum_{R\subseteq S, R\not=\emptyset} \Delta_v(T\setminus \{[S]\}\cup R)\\
&=\sum_{T\subseteq N\setminus S\cup\{[S]\}, [S]\in T}\sum_{R\subseteq S, R\not=\emptyset}p^{|R|}_{|T\setminus\{[S]\}\cup R|} \Delta_v(T\setminus \{[S]\}\cup R)\\
&=\sum_{T\subseteq N\setminus S}\sum_{R\subseteq S, R\not=\emptyset}p^{|R|}_{|T\cup R|} \Delta_v(T\cup R)=\sum_{T\subseteq N, S\cap T\not= \emptyset}p^{|S\cap T|}_{|T|} \Delta_v(T)=\varphi_S(N,v).
\end{align*}
\end{proof}
\end{lemma}

\begin{lemma}\label{lemma:extending_semivalues}
Let $\varphi$ be a group weak consistent semivalue with weights $(p^q_t)_{q,t\in\mathbb{N},q\leq t}$ and $\varphi'$ a player consistent semivalue with weights $(p'_t)_{t\in\mathbb{N}}$ from \cref{lemma:semivalue_npo_divideds_formula}. It holds that $\varphi$ is an extension of $\varphi'$ if and only if $p^1_t=p'_t$ for every $t\in\mathbb{N}$.

\begin{proof}
$\Rightarrow:$ 
Let $\varphi$ be an arbitrary group value that extends $\varphi'$.
Fix an arbitrary $t\in\mathbb{N}$ and take any set of players $N$ and $S \subseteq N$ such that $|S|=t$.
From the fact that $\varphi$ extends $\varphi'$ we obtain
\[
p^1_t=\sum_{T\subseteq N, i\in T}p^1_{|T|}\Delta_{u_S}(T)=\varphi_{\{i\}}(N,u_S)=\varphi'_i(N,u_S)=\sum_{T\subseteq N, i\in T}p'_{|T|}\Delta_{u_S}(T)=p'_t
\]

$\Leftarrow:$ For every game $(N,v)$ and player $i\in N$ we get
\[
\varphi_{\{i\}}(N,v)=\sum_{T\subseteq N, i\in T}p^1_{|T|}\Delta_{v}(T)=\sum_{T\subseteq N, i\in T}p'_{|T|}\Delta_{v}(T)=\varphi_i(N,v)
\]
\end{proof}
\end{lemma}

\subsection*{Proof of \cref{theorem:semivalues_player_extensions}}
\begin{proof}
Let $\varphi'$ be the player-consistent semivalue with weights $(p'_t)_{t \in \mathbb{N}}$ from \cref{lemma:semivalue_npo_divideds_formula}.
It is clear that the sum, merge, and sequential extensions are unique if they exist, since they are all defined by formulas based on player values.
Hence, it remains to show that each of these extensions exists and is a group weakly consistent semivalue.

\paragraph{Sum extension}
Let $\varphi$ be a group weak consistent semivalue with weights $p^q_t=q\cdot p'_t$. From \cref{lemma:extending_semivalues}, $\varphi$ extends $\varphi'$ since $p^1_t=1\cdot p'_t=p'_t$ for every $t\in\mathbb{N}$. 

We can directly show that $\varphi$ is a sum extension. For every game $(N,v)$ and coalition $S\subseteq N$
\begin{align*}
\varphi_S(N,v)
&=\sum_{T\subseteq N, S\cap T\not=\emptyset}p^{|S\cap T|}_{|T|}\Delta_v(T)
=\sum_{T\subseteq N, S\cap T\not=\emptyset}|S\cap T|\cdot p'_{|T|}\Delta_v(T)\\
&=\sum_{i\in S}\sum_{T\subseteq N, i\in T}p'_{|T|}\Delta_v(T)=\sum_{i\in S}\varphi'_i(N,v)
\end{align*}

\paragraph{Merge extension}
Let $\varphi$ be a group weak consistent semivalue with weights $p^q_t=p'_{t-q+1}$. From \cref{lemma:extending_semivalues}, $\varphi$ extends $\varphi'$ as it holds that $p^1_t=p'_{t-1+1}=p'_t$ for every $t\in\mathbb{N}$. 
From \cref{lemma:gws:merge_weights}, we get that $\varphi$ is a merge extension since $p^q_t=p'_{t-q+1}=p'_{(t+1)-(q+1)+1}=p^{q+1}_{t+1}$.

\paragraph{Sequential extension}
Let $\varphi$ be a group weakly consistent semivalue with weights $p_t^q = p'_t$.
By \cref{lemma:extending_semivalues}, $\varphi$ extends $\varphi'$, since $p_t^1 = p'_t$ for every $t \in \mathbb{N}$.
To prove that $\varphi$ is a sequential extension, we show that it satisfies Balanced Contributions, which, together with \cref{theorem:sequential_balanced_axiom}, proves the thesis.
For every game $(N,v)$ and coalitions $S, Q \subseteq N$, it holds that

\begin{align*}
\varphi_S(N,v)-\varphi_{S\setminus Q}(N\setminus Q,v)
&=\sum_{T\subseteq N, S\cap T\not=\emptyset}p^{|S\cap T|}_{|T|}\Delta_v(T)-\sum_{T\subseteq N\setminus Q, (S\setminus Q)\cap T\not=\emptyset}p^{|(S\setminus Q)\cap T|}_{|T|}\Delta_v(T)\\
&=\sum_{T\subseteq N, S\cap T\not=\emptyset}p'_{|T|}\Delta_v(T)-\sum_{T\subseteq N\setminus Q, S\cap T\not=\emptyset}p'_{|T|}\Delta_v(T)\\
&=\sum_{T\subseteq N, S\cap T\not=\emptyset, Q\cap T\not=\emptyset}p'_{|T|}\Delta_v(T)\\
&=\sum_{T\subseteq N, Q\cap T\not=\emptyset}p'_{|T|}\Delta_v(T)-\sum_{T\subseteq N\setminus S, Q\cap T\not=\emptyset}p'_{|T|}\Delta_v(T)\\
&=\sum_{T\subseteq N, Q\cap T\not=\emptyset}p^{|Q\cap T|}_{|T|}\Delta_v(T)-\sum_{T\subseteq N\setminus S, (Q\setminus S)\cap T\not=\emptyset}p^{|(Q\setminus S)\cap T|}_{|T|}\Delta_v(T)\\
&=\varphi_Q(N,v)-\varphi_{Q\setminus S}(N\setminus S,v)
\end{align*}

\end{proof}

\subsection*{Proof of \cref{proposition:merge_axioms}}
\begin{proof}
Every group weak consistent semivalue is of the form
\begin{equation*}
\varphi_S(N, v) = \sum_{T \subseteq N, S \cap T\not=\emptyset}p^{|S\cap T|}_{|T|} \Delta_v(T).
\end{equation*}
for some collection of weights $(p^q_t)_{q,t\in\mathbb{N};q\leq t}$.

From \cite{Marichal:etal:2007}, every generalized value is of the form
\begin{equation*}
\varphi_S(N, v) = \sum_{T \subseteq N, S \cap T\not=\emptyset}\hat{p}^{|S|}_{|T\setminus S|}\Delta_v(T).
\end{equation*}
for some collection of weights $(\hat{p}^q_t)_{q,t\in[n];q\leq t}$.

Now, for every $q,t\in\mathbb{N}, q\leq t$ take any sets of players $S, Q$ and $N$ such that $S\not\subseteq Q$, $|S\cap Q|=q$, $|Q|=t$, and $S,Q\subseteq N$. Let $i$ be a player such that $i\in S \setminus Q$. Then
\begin{align*}
p^q_t=p^{|S\cap Q|}_{|Q|}
&=\sum_{T \subseteq N, S \cap T\not=\emptyset}p^{|S\cap T|}_{|T|} \Delta_{u_Q}(T)=\varphi_S(N,u_Q)\\
&=\sum_{T \subseteq N, S \cap T\not=\emptyset}\hat{p}^{|S|}_{|T\setminus S|}\Delta_{u_Q}(T)=\hat{p}^{|S|}_{|Q\setminus S|}\\
&=\hat{p}^{|S|}_{|(Q\cup\{i\})\setminus S|}=\sum_{T \subseteq N, S \cap T\not=\emptyset}\hat{p}^{|S|}_{|T\setminus S|}\Delta_{u_{Q\cup\{i\}}}(T)\\
&=\varphi_S(N,u_{Q\cup\{i\}})=\sum_{T \subseteq N, S \cap T\not=\emptyset}p^{|S\cap T|}_{|T|}\Delta_{u_{Q\cup\{i\}}}(T)\\
&=p^{|S\cap (Q\cup\{i\})|}_{|Q\cup\{i\}|}=p^{q+1}_{t+1}.
\end{align*}
This proves that $\varphi$ is a merge extension directly from \cref{lemma:gws:merge_weights}.
\end{proof}

\subsection*{Proof of \cref{proposition:efficiency_equivalence}}
\begin{proof}
First we will show the equivalence between $\varphi$ extending the Shapley value and having its weights satisfy $p^1_t=1/t$.

Let $(p'_t)_{t\in\mathbb{N}}$ be the weights of the Shapley value in terms of \cref{lemma:semivalue_npo_divideds_formula}. 
Recall that $SV_i(N,v)=\sum_{i\in T}\frac{\Delta_v(T)}{|T|}$, so it follows that $p'_t=1/t$.
From \cref{lemma:extending_semivalues}, $\varphi$ extends the Shapley value if and only if $p^1_t=p'_t=1/t$, proving the equivalence.

Now we will show the equivalence between $\varphi$ satisfying Singleton-Efficiency and having its weights satisfy $p^1_t=1/t$.
With this restriction on weights it satisfies Singleton-Efficiency, since then for every game $(N,v)$ 
\[
\sum_{i\in N}\varphi_{\{i\}}(N,v)=\sum_{i\in N}\sum_{T\subseteq N,i\in T}p^{|\{i\}\cap T|}_{|T|}\Delta_v(T)=\sum_{i\in N}\sum_{T\subseteq N,i\in T}\frac{1}{|T|}\Delta_v(T)=\sum_{T\subseteq N}\Delta_v(T)=v(N).
\]
Now we will show that Singleton-Efficiency implies those weight restrictions. 
Recall that from \cref{theorem:semivalues}, $\varphi$ satisfies Dummy Player and Symmetry. 
For every $t\in \mathbb{N}$ consider player sets $S, N$ and player $i\in S$, such that $|S|=t$ and $S\subseteq N$. 
In a game $u_S$ all players outside of $S$ are null players, and since $\varphi$ satisfies Dummy Player their singletons' values are 0. 
All players in $S$ are symmetric in this game, and since $\varphi$ satisfies Symmetry, their singletons' values are be equal. 
Since $\varphi$ satisfies Singleton-Efficiency, the sum of all singletons' values are equal to $u_S(N)=1$. 
Then all singletons of players in $S$ receive $\frac{1}{|S|}=\frac{1}{t}$ and
\[
\frac{1}{t}=\varphi_{\{i\}}(N,u_S)=\sum_{T\subseteq N,S\cap T\not=\emptyset}p^{|\{i\}\cap T|}_{|T|}\Delta_{u_S}(T)=p^1_{|S|}=p^1_t,
\]
concluding the proof.
\end{proof}

\subsection*{Proof of \cref{theorem:axioms_union}}
\begin{proof}
$\Rightarrow:$
The Union Shapley value is a group weak consistent semivalue, so from \cref{theorem:semivalues} it satisfies Linearity, Symmetry, Dummy Player, Weak Monotonicity and Null Player Out. 
Since it extends the Shapley value, from \cref{proposition:efficiency_equivalence} it satisfies Singleton-Efficiency. 
It remains to show it satisfies Group Equality.

For every game $(N,v)$ and non-empty coalitions $S,Q\subseteq N$ using \cref{equation:union_shapley} we obtain
\[
US_S(N,u_N)=\sum_{T\subseteq N, S\cap T\not=\emptyset}\frac{\Delta_{u_N}(T)}{|T|}=\frac{1}{|N|}=\sum_{T\subseteq N, Q\cap T\not=\emptyset}\frac{\Delta_{u_N}(T)}{|T|}=US_Q(N,u_N)
\]
$\Leftarrow:$
Let $\varphi$ be a group value satisfying these axioms. 
it satisfies Linearity, Symmetry, Dummy Player, Weak Monotonicity and Null Player Out it is a group weak consistent semivalue from \cref{theorem:semivalues}. 
Since it satisfies Singleton-Efficiency, from \cref{proposition:efficiency_equivalence} it holds that weights of $\varphi$ satisfy $p^1_t=1/t$. Now we will show that all weights satisfy $p^q_t=1/t$, proving $\varphi$ is the Union Shapley value from \cref{theorem:semivalues}. For every $q,t\in \mathbb{N}$ such that $q\leq t$ choose sets of players $S,N$ and player $i\in N$ such that $S\subseteq N,|S|=q, |N|=t$. Then using Group Equality we get
\begin{align*}
p^q_t
&=p^{|S\cap N|}_{|N|}=\sum_{T\subseteq N, S\cap T\not=\emptyset}p^{|S\cap T|}_{|T|}\Delta_{u_N}(T)=\varphi_S(N,u_N)\\
&=\varphi_{\{i\}}(N,u_N)=\sum_{T\subseteq N, \{i\}\cap T\not=\emptyset}p^{|\{i\}\cap T|}_{|T|}\Delta_{u_N}(T)=p^{|\{i\}\cap N|}_{|N|}=p^1_t=1/t.
\end{align*}
\end{proof}

\subsection*{Proof of \cref{theorem:axioms_sum}}
\begin{proof}
$\Rightarrow:$
The sum extension of the Shapley value is a group weak consistent semivalue from \cref{theorem:semivalues_player_extensions}, so from \cref{theorem:semivalues} it satisfies Linearity, Symmetry, Dummy Player, Weak Monotonicity and Null Player Out. 
Since it extends the Shapley value, from \cref{proposition:efficiency_equivalence} it satisfies Singleton-Efficiency. It remains to show it satisfies Group Proportionality. Let $\varphi$ be the sum extension of the Shapley value.

For every game $(N,v)$ and non-empty coalitions $S,T\subseteq N$ we have
\begin{align*}
\frac{\varphi_S(N,u_N)}{|S|}&=\frac{1}{|S|}\sum_{i\in S}SV_i(N,u_N)=\frac{1}{|S|}\frac{|S|}{|N|}=\frac{1}{|N|}\\
&=\frac{1}{|T|}\frac{|T|}{|N|}=\frac{1}{|T|}\sum_{i\in T}SV_i(N,u_N)=\frac{\varphi_T(N,u_N)}{|T|}
\end{align*}
$\Leftarrow:$
Let $\varphi$ be a group value satisfying these axioms. Since it satisfies Linearity, Symmetry, Dummy Player, Weak Monotonicity and Null Player Out it is a group weak consistent semivalue from \cref{theorem:semivalues}. 
Since it satisfies Singleton-Efficiency, from \cref{proposition:efficiency_equivalence} it holds that weights of $\varphi$ satisfy $p^1_t=1/t$. Now we will show that all weights satisfy $p^q_t=q/t$, proving $\varphi$ is the sum extension of the Shapley value from \cref{theorem:semivalues}. For every $q,t\in \mathbb{N}$ such that $q\leq t$ choose sets of players $S,N$ and player $i\in N$ such that $S\subseteq N,|S|=q, |N|=t$. Then using Group Proportionality we get
\begin{align*}
p^q_t
&=p^{|S\cap N|}_{|N|}=\sum_{T\subseteq N, S\cap T\not=\emptyset}p^{|S\cap T|}_{|T|}\Delta_{u_N}(T)=\varphi_S(N,u_N)\\
&=|S|\cdot\varphi_{\{i\}}(N,u_N)=|S|\cdot \sum_{T\subseteq N, \{i\}\cap T\not=\emptyset}p^{|\{i\}\cap T|}_{|T|}\Delta_{u_N}(T)=|S|\cdot p^{|\{i\}\cap N|}_{|N|}=q\cdot p^1_t=q/t.
\end{align*}
\end{proof}

\subsection*{Proof of \cref{theorem:synergistic_semivalues}}
\begin{proof}
$\Rightarrow:$ Let $\varphi$ be a group value satisfying Linearity, Symmetry, Dummifying Player, Weak Monotonicity and Null Player Out. From \cref{theorem:symmetry} it holds that for a fixed $N$, $\varphi$ is of the form
\[
\varphi_S(N, v) = \sum_{T\subseteq N, T\not=\emptyset}p^{|S|, |S\cap T|}_{|T|} \Delta_v(T).
\]
Moreover, for every $S, T\subseteq N$, we get $p^{|S|, |S\cap T|}_{|T|} = \varphi_S(N, u_T)$.

First, fix $N$ and note that for every coalition $S$ and $T$ such that $S\not\subseteq T$ there is at least one player $i$ such that $i\in S$ and $i\not\in T$. 
In game $u_T$, player $i$ is a null player which from Dummifying Player implies that the value of coalition $S$ is also zero:
\[
\varphi_S(N, u_T) = p^{|S|, |S\cap T|}_{|T|} = 0.
\]
This concludes the proof that only dividends for which $S\subseteq T$ have non-zero weights. 
Moreover, notice that for $S\subseteq T$, $S\cap T=S$, which means that the information about the size of the intersection is redundant. 
All together gives us the form:
\[
\varphi_S(N, v)=\sum_{S\subseteq T \subseteq N}p^{|S|}_{|T|} \Delta_v(T)
\]
In particular, for every $S, T\subseteq N$ such that $S\subseteq T$, we get $p^{|S|}_{|T|}=\varphi_S(N, u_T)$.

It remains to prove that weights are independent of $N$, $p^1_1=1$ and that all weights are non-negative.
Let $(p^q_t(N))_{q,t\in [n],q \leq t}$ be the collection of weights of $\varphi$ for a given $N$. For every set of players $N$ and player $i\not\in N$, consider a game $(N\cup\{i\},v)$, where $i$ is a null player. Then for every $q,t\in[n], q\leq t$ take any coalitions $S,Q\subseteq N$, such that $S\subseteq Q, |S|=q$ and $|Q|=t$ and through Null Player Out it holds that
\begin{align*}
p^q_t(N)=p^{|S|}_{|Q|}(N)=\varphi_S(N,u_Q)=\varphi_S(N\cup\{i\},u_Q)=p^{|S|}_{|Q|}(N\cup\{i\})=p^q_t(N\cup\{i\}),
\end{align*}
which shows that the weights do not change when adding a new player for a game.

Fix another set of players $N'$. Since we have just argued that adding a player does not change the weights, by repeatedly adding players we obtain
\[
p_t^q(N) = p_t^q(N \cup N') = p_t^q(N'),
\]
which shows that the weights do not depend on the set of players.

Next, for every $q,t\in\mathbb{N}$ with $q\leq t$, fix sets of players $S,T,\subseteq N$ and player $i\in N$, such that $S\subseteq T$, $|S|=q$, $|T|=t$. Since the game $(N,u_T)$ is positive, Weak Monotonicity gives
\[
\varphi_S(N,u_T)=p^{|S|}_{|T|}=p^q_t\geq 0. 
\]
Furthermore, in the game $(N,u_{\{i\}})$, player $i$ is a dummy player, so
\[
\varphi_{\{i\}}(N,u_{\{i\}})=p^{|\{i\}|}_{|\{i\}|}=p^1_1=1.
\]
Hence, $\varphi$ is a synergy semivalue.

$\Leftarrow:$ 
We know that $\varphi$ satisfies Linearity and Symmetry from \cref{theorem:symmetry}. 
Let us focus on Dummifying Player, Weak Monotonicity and Null Player Out.

It is known that for every game $(N, v)$ and dummy player $i\in N$, $\Delta_v(\{i\})=v(\{i\})$ and for all coalitions $\{i\}\subsetneq S$, $\Delta_v(S)=0$. 
This means that:
\[
\varphi_{\{i\}}(N, v)=\sum_{\{i\}\subseteq T\subseteq N}p^{|\{i\}|}_{|T|}\Delta_v(T)=p^1_1\Delta_v(\{i\})=v(\{i\})
\]
and also for every non-singleton coalition $S$ such that $i \in S$:
\[
\varphi_S(N, v)=\sum_{S\subseteq T}p^{|S|}_{|T|}\Delta_v(T)=0.
\]
This proves that $\varphi$ satisfies Dummifying Player.

We know that all the weights are non-negative. For every positive game $(N,v)$ we know that the dividends are also non-negative. This means that for every coalition $S\subseteq N$ it holds that:
\[
\varphi_S(N, v)=\sum_{S\subseteq T}p^{|S|}_{|T|}\Delta_v(T)\geq 0,
\]
This concludes the proof that $\varphi$ satisfies also Weak Monotonicity.

We also directly show that $\varphi$ satisfies Null Player Out. For every game $(N,v)$, null player $i\in N$ and non-empty coalition $S$ we have
\[
\varphi_S(N,v)=\sum_{S\subseteq T\subseteq N}p^{|S\cap T|}_{|T|}\Delta_v(T)=\sum_{S\subseteq T\subseteq N\setminus\{i\}}p^{|S\cap T|}_{|T|}\Delta_v(T)=\varphi_S(N\setminus\{i\},v),
\]
since, as you may recall, all dividends for coalitions containing a null player are equal to zero.
\end{proof}

\begin{lemma}\label{lemma:efficiency_synergy}
If a synergy semivalue satisfies Singleton-Efficiency, its weights satisfy $p^1_t=1/t$.
\begin{proof}
Let $\varphi$ be a synergy semivalue satisfying Singleton-Efficiency.
Recall that from \cref{theorem:synergistic_semivalues}, $\varphi$ satisfies Dummifying Player and Symmetry. For every $t\in \mathbb{N}$ consider player sets $S, N$ and player $i\in S$, such that $|S|=t$ and $S\subseteq N$. In a game $u_S$ all players outside of $S$ are null players, and since $\varphi$ satisfies Dummifying Player their singletons' values are 0. All players in $S$ are symmetric in this game, and since $\varphi$ satisfies Symmetry, their singletons' values are equal. Since $\varphi$ satisfies Singleton-Efficiency, the sum of all singletons' values are equal to $u_S(N)=1$. Then all singletons of players in $S$ receive $\frac{1}{|S|}=\frac{1}{t}$ and
\[
\frac{1}{t}=\varphi_{\{i\}}(N,u_S)=\sum_{\{i\}\subseteq T\subseteq N}p^{|\{i\}|}_{|T|}\Delta_{u_S}(T)=p^1_{|S|}=p^1_t,
\]
concluding the proof.
\end{proof}
\end{lemma}

\subsection*{Proof of \cref{proposition:axioms_intersection}}
\begin{proof}
$\Rightarrow:$
The Intersection Shapley value is a synergy semivalue, so from \cref{theorem:synergistic_semivalues} it satisfies Linearity, Symmetry, Dummifying Player, Weak Monotonicity and Null Player Out. 
First we show it satisfies Singleton-Efficiency. 
For every game $(N,v)$
\[
\sum_{i\in N}IS_{\{i\}}(N,v)=\sum_{i\in N}SV_i(N,v)=v(N).
\]
It remains to show it satisfies Group Equality.

For every game $(N,v)$ and non-empty coalitions $S,Q\subseteq N$ using \cref{equation:intersection_shapley} we obtain
\[
IS_S(N,u_N)=\sum_{S\subseteq T\subseteq N}\frac{\Delta_{u_N}(T)}{|T|}=\frac{1}{|N|}=\sum_{Q\subseteq T\subseteq N}\frac{\Delta_{u_N}(T)}{|T|}=IS_Q(N,u_N)
\]
$\Leftarrow:$
Let $\varphi$ be a group value satisfying these axioms. Since it satisfies Linearity, Symmetry, Dummifying Player, Weak Monotonicity and Null Player Out it is a synergy semivalue from \cref{theorem:synergistic_semivalues}, and let its weights be $(p^q_t)_{p,q\in\mathbb{N}, q\leq t}$. 
Since $\varphi$ satisfies Singleton-Efficiency, from \cref{lemma:efficiency_synergy} it holds that $p^1_t=1/t$. 
Now we will show that all weights satisfy $p^q_t=1/t$, proving $\varphi$ is the Intersection Shapley value from \cref{theorem:synergistic_semivalues}. 
For every $q,t\in \mathbb{N}$ such that $q\leq t$ choose sets of players $S,N$ and player $i\in N$ such that $S\subseteq N,|S|=q, |N|=t$. Then using Group Equality we get
\begin{align*}
p^q_t
&=p^{|S|}_{|N|}=\sum_{S\subseteq T\subseteq N}p^{|S|}_{|T|}\Delta_{u_N}(T)=\varphi_S(N,u_N)\\
&=\varphi_{\{i\}}(N,u_N)=\sum_{\{i\}\subseteq T\subseteq N}p^{|\{i\}|}_{|T|}\Delta_{u_N}(T)=p^{|\{i\}|}_{|N|}=p^1_t=1/t.
\end{align*}
\end{proof}

\subsection*{Proof of \cref{proposition:axioms_intersection_s}}
\begin{proof}
$\Rightarrow:$
The group value $|S|\cdot IS_S$ is a synergy semivalue, so from \cref{theorem:synergistic_semivalues}it satisfies Linearity, Symmetry, Dummifying Player, Weak Monotonicity and Null Player Out. 
First we show it satisfies Singleton-Efficiency. 
For every game $(N,v)$
\[
\sum_{i\in N}|\{i\}|\cdot IS_{\{i\}}(N,v)=\sum_{i\in N}SV_i(N,v)=v(N).
\]
It remains to show it satisfies Group Proportionality.

For every game $(N,v)$ and non-empty coalitions $S,Q\subseteq N$ using \cref{equation:intersection_shapley} we obtain
\[
\frac{|S|\cdot IS_S(N,u_N)}{|S|}=\sum_{S\subseteq T\subseteq N}\frac{\Delta_{u_N}(T)}{|T|}=\frac{1}{|N|}=\sum_{Q\subseteq T\subseteq N}\frac{\Delta_{u_N}(T)}{|T|}=\frac{|Q|\cdot IS_Q(N,u_N)}{|Q|}
\]
$\Leftarrow:$
Let $\varphi$ be a group value satisfying these axioms. Since it satisfies Linearity, Symmetry, Dummifying Player, Weak Monotonicity and Null Player Out it is a synergy semivalue from \cref{theorem:synergistic_semivalues}, and let its weights be $(p^q_t)_{p,q\in\mathbb{N}, q\leq t}$. 
Since $\varphi$ satisfies Singleton-Efficiency, from \cref{lemma:efficiency_synergy} it holds that $p^1_t=1/t$. 
Now we will show that all weights satisfy $p^q_t=q/t$, proving $\varphi_S=|S|\cdot IS_S$ from \cref{theorem:synergistic_semivalues}. For every $q,t\in \mathbb{N}$ such that $q\leq t$ choose sets of players $S,N$ and player $i\in N$ such that $S\subseteq N,|S|=q, |N|=t$. Then using Group Proportionality we get
\begin{align*}
p^q_t
&=p^{|S|}_{|N|}=\sum_{S\subseteq T\subseteq N}p^{|S|}_{|T|}\Delta_{u_N}(T)=\varphi_S(N,u_N)\\
&=|S|\cdot\varphi_{\{i\}}(N,u_N)=|S|\cdot \sum_{\{i\}\subseteq T\subseteq N}p^{|\{i\}|}_{|T|}\Delta_{u_N}(T)=|S|\cdot p^{|\{i\}|}_{|N|}=q\cdot p^1_t=q/t.
\end{align*}
\end{proof}

\subsection*{Proof of \cref{proposition:union_intersection_pair}}
\begin{proof}
Since $\varphi$ and $\hat{\varphi}$ are corresponding group weak consistent semivalue and synergy semivalue, they share the same weights $(p^q_t)_{q,t\in\mathbb{N},q\leq t}$. They extend the Shapley value, so from \cref{proposition:efficiency_equivalence} it holds that $p^1_t=1/t$. For every $t\in \mathbb{N}, t\geq 2$ take any player sets $S, N$ and players $i,j\in S$, such that $|S|=t$, $S\subseteq N$ and $i\not=j$. Then
\begin{align*}
p^2_t
&=\frac{1}{2}\cdot 2p^{|\{i,j\}\cap S|}_{|S|}=\frac{1}{2}(\sum_{T\subseteq N, \{i,j\}\cap T\not=\emptyset}p^{|\{i,j\}\cap T|}_{|T|}\Delta_{u_S}(T)+\sum_{\{i,j\}\subseteq T\subseteq N}p^{|\{i,j\}\cap T|}_{|T|}\Delta_{u_S}(T))\\
&=\frac{1}{2}(\varphi_{\{i,j\}}(N,u_S)+\hat{\varphi}_{\{i,j\}}(N,u_S))=\frac{1}{2}(SV_i(N,u_S)+SV_j(N,u_S))=\frac{1}{2}\cdot \frac{2}{|S|}=\frac{1}{t}.
\end{align*}
Now we get that for every game $(N,v)$ and players $i,j\in N, i\not=j$
\[
\varphi_{\{i,j\}}(N,v)=\sum_{T\subseteq N, \{i,j\}\cap T\not=\emptyset}p^{|\{i,j\}\cap T|}_{|T|}\Delta_v(T)=\sum_{T\subseteq N, \{i,j\}\cap T\not=\emptyset}\frac{\Delta_v(T)}{|T|}=US_{\{i,j\}}(N,v)
\]
and
\[
\hat{\varphi}_{\{i,j\}}(N,v)=\sum_{\{i,j\}\subseteq T\subseteq N}p^{|\{i,j\}\cap T|}_{|T|}\Delta_v(T)=\sum_{\{i,j\}\subseteq T\subseteq N}\frac{\Delta_v(T)}{|T|}=IS_{\{i,j\}}(N,v),
\]
concluding the proof.
\end{proof}

\subsection*{Proof of \cref{theorem:usisrelation}}
\begin{proof}
Using the dividend formula for Intersection Shapley value, we get
\begin{align*}
\sum_{\emptyset \subsetneq T \subseteq S} (-1)^{|T|-1} IS_T(N,v) & = \sum_{T: \emptyset\subsetneq T\subseteq S}(-1)^{|T|-1}\sum_{Q: T\subseteq Q}\frac{\Delta_v(Q)}{|Q|} \\
& = \sum_{Q: Q\cap S\not=\emptyset}\frac{\Delta_v(Q)}{|Q|}\sum_{T:\emptyset\subsetneq T\subseteq Q}(-1)^{|T|-1}.
\end{align*}
Note that:
\[
\sum_{T:\emptyset\subsetneq T\subseteq Q}(-1)^{|T|-1} = - ((-1+1)^{|Q|} - 1) = 1.
\]
Combining both equations, we obtain
\[
\sum_{\emptyset \subsetneq T \subseteq S} (-1)^{|T|-1} IS_T(N,v) = \sum_{Q:Q\cap S\not=\emptyset}\frac{\Delta_v(Q)}{|Q|} \cdot 1 = US_S(N, v),
\]
which concludes the proof.
\end{proof}

\end{document}